\documentclass[a4paper,english,cleveref,autoref]{lipics-v2019}

\usepackage{amsmath,amssymb}
\usepackage[ruled,vlined,linesnumbered]{algorithm2e}
\usepackage{tikz}
\usepackage{comment}
\usepackage{thmtools}
\usepackage{thm-restate}
\usepackage{xspace}

\tikzstyle{vertex}=[circle,fill=black!25,minimum size=20pt,inner sep=0pt]
\tikzstyle{adversary}=[circle,fill=black!60,minimum size=20pt,inner sep=0pt]

\theoremstyle{plain}
\newtheorem{apxlemma}{Lemma}[section]
\crefname{apxlemma}{Lemma}{Lemmas}

\newtheorem{assumption}{Assumption}{\bfseries}{\upshape}
\crefname{assumption}{Assumption}{Assumptions}

\newcommand{\gro}{G_1} 
\newcommand{\grt}{G_2} 

\newcommand{\LINE}{L} 
\newcommand{\PLN}{\psi} 
\newcommand{\PL}[1]{\PLN(#1)} 

\newcommand{\HEAVY}{W_{\text{H}}} 
\newcommand{\LIGHT}{W_{\text{L}}} 

\newcommand{\opto}{O_1} 
\newcommand{\optt}{O_2} 

\newcommand{\algo}{M_1} 
\newcommand{\algt}{M_2} 

\newcommand{\core}{\overline{M}}

\newcommand{\g}[2][\opto]{\operatorname{gain}_{#1}(#2)} 

\newcommand{\child}[1]{\operatorname{ch}(#1)} 
\newcommand{\ancestor}[1]{\operatorname{anc}(#1)} 
\newcommand{\descendant}[1]{\operatorname{desc}(#1)} 

\newcommand{\algKnownSto}{Algorithm~1.1\xspace} 
\newcommand{\algKnownStt}{Algorithm~1.2\xspace} 

\newcommand{\algUnknownSto}{Algorithm~2.1\xspace} 
\newcommand{\algUnknownStt}{Algorithm~2.2\xspace}

\usepackage{tikz} 
\usetikzlibrary{calc,positioning,shapes.multipart}
\usetikzlibrary{decorations.pathreplacing}
\usetikzlibrary{shapes.misc}

\definecolor{TUMBlau}{RGB}{0,82,147} 
\definecolor{TUMRot}{RGB}{202,033,063}
 
\tikzset{cross/.style={cross out, draw, minimum size=2*(#1-\pgflinewidth), inner sep=0pt, outer sep=0pt}}
\tikzstyle{vertex}=[circle,fill=black!25,minimum size=20pt,inner sep=0pt]
\tikzstyle{adversary}=[cross,color=black!60,ultra thick, minimum size=20pt,inner sep=0pt]
\usetikzlibrary{calc}
\tikzset{extremely thick/.style={line width=1.5mm}}

\usepackage{ifthen}
\usepackage[textsize=scriptsize]{todonotes}
\newboolean{showNotes}
\setboolean{showNotes}{true}
\newcommand{\ulrike}[1]{\ifthenelse{\boolean{showNotes}} {\todo[color=yellow]{#1}}{}
}
\newcommand{\jannik}[1]{\ifthenelse {\boolean{showNotes}}{\todo[color=blue!30]{#1}}{}
} 
\newcommand{\jose}[1]{\ifthenelse {\boolean{showNotes}}{\todo[color=green!30]{#1}}{}
} 

\nolinenumbers

\ccsdesc[300]{Mathematics of computing~Combinatorial algorithms}

\title{Maintaining Perfect Matchings at Low Cost} 
\titlerunning{Maintaining Perfect Matchings at Low Cost}

\author{Jannik Matuschke}{Research Center for Operations Management, KU Leuven, Belgium}{jannik.matuschke@kuleuven.be}{}{}
\author{Ulrike Schmidt-Kraepelin}{Institute of Software Engineering and Theoretical Computer Science, TU Berlin, Germany}{u.schmidt-kraepelin@tu-berlin.de}{}{}
\author{Jos\'e Verschae}{Institute of Engineering Sciences, Universidad de O'Higgins, Chile}{jose.verschae@uoh.cl}{}{}

\authorrunning{J.\,Q. Public and J.\,R. Public}

\authorrunning{J.\,Matuschke, U.\,Schmidt-Kraepelin and J.\,Verschae}

\Copyright{Jannik Matuschke, Ulrike Schmidt-Kraepelin and Jos\'e Verschae}

\keywords{matchings, robust optimization, approximation algorithms}

\category{}

\supplement{}

\funding{The authors were supported by Project Fondecyt Nr. 1181527, the Alexander von Humboldt Foundation with funds of the German Federal Ministry of Education and Research (BMBF), BAYLAT, and the Deutsche Forschungsgemeinschaft under grant BR 4744/2-1.}

\acknowledgements{We thank anonymous reviewers for their helpful feedback.}

\hideLIPIcs  

\EventEditors{Christel Baier, Ioannis Chatzigiannakis, Paola Flocchini, and Stefano Leonardi}
\EventNoEds{4}
\EventLongTitle{46th International Colloquium on Automata, Languages, and Programming (ICALP 2019)}
\EventShortTitle{ICALP 2019}
\EventAcronym{ICALP}
\EventYear{2019}
\EventDate{July 9--12, 2019}
\EventLocation{Patras, Greece}
\EventLogo{eatcs}
\SeriesVolume{132}
\ArticleNo{XXX}

\begin{document}

\maketitle

\begin{abstract}
The min-cost matching problem suffers from being very sensitive to small changes of the input. Even in a simple setting, e.g., when the costs come from the metric on the line, adding two nodes to the input might change the optimal solution completely. On the other hand, one expects that small changes in the input should incur only small changes on the constructed solutions, measured as the number of modified edges. We introduce a two-stage model where we study the trade-off between quality and robustness of solutions. In the first stage we are given a set of nodes in a metric space and we must compute a perfect matching. In the second stage $2k$ new nodes appear and we must adapt the solution to a perfect matching for the new instance. 

We say that an algorithm is $(\alpha,\beta)$-robust if the solutions constructed in both stages are $\alpha$-approximate with respect to min-cost perfect matchings, and if the number of edges deleted from the first stage matching is at most~$\beta k$. Hence, $\alpha$ measures the quality of the algorithm and $\beta$ its robustness. In this setting we aim to balance both measures by deriving algorithms for constant $\alpha$ and $\beta$.
We show that there exists an algorithm that is $(3,1)$-robust for any metric if one knows the number $2k$ of arriving nodes in advance. For the case that $k$ is unknown the situation is significantly more involved. We study this setting under the metric on the line and devise a $(10,2)$-robust algorithm that constructs a solution with a recursive structure that carefully balances cost and redundancy. 
\end{abstract}

\section{Introduction}

Weighted matching is one of the founding problems in combinatorial optimization, playing an important role in the settling of the area. The work by Edmonds~\cite{1965:Edmonds} on this problem greatly influenced the role of polyhedral theory on algorithm design~\cite{2012:Pulleyblank}. On the other hand, the problem found applications in several domains~\cite{1983:Ball,1994:Bell,1992:Derigs,1990:Olafsson,1981:Reingold}. In particular routing problems are an important area of application, and its procedures often appeared as subroutines of other important algorithms, the most notable being Christofides' algorithm~\cite{1976:Christofides} for the \textit{traveling salesperson problem (TSP)}. \newpage

An important aspect of devising solution methods for optimization problems is studying the sensitivity of the solution towards small changes in the input. This sensitivity analysis has a long history and plays an important role in practice~\cite{1977:Geoffrion}. Min-cost matching is a problem that has particularly sensitive optimal solutions. Assume for example that nodes lie on the real line at points $\ell$ and $\ell+1-\varepsilon$ for some $0<\varepsilon<1$ and all $\ell\in\{1,\ldots,n\}$, see \cref{fig:motivationKKnown}. The min-cost matching, for costs equal the distance on the line, is simply the edges $\{\ell,\ell+1-\varepsilon\}$. However, even under a minor modification of the input, e.g., if two new nodes appear at points $1-\varepsilon$ and $n+1$, the optimal solution changes all of its edges, and furthermore the cost decreases by a $\Theta(1/\varepsilon)$ factor.
Rearranging many edges in an existing solution is often undesirable and may incur large costs, for example in an application context where the matching edges imply physical connections or binding commitments between nodes. A natural question in this context is whether we can avoid such a large number of rearrangements by constructing a \emph{robust} solution that is only slightly more expensive. In other words, we are interested in studying the trade-off between robustness and the cost of solutions.

\begin{figure}[t]
\centering
\begin{tikzpicture}
  \node[vertex,minimum size=4](o1) at (0,0) {};
  \node[vertex,minimum size=4, right = 1 of o1](o2){};
  \node[vertex,minimum size=4, right = 0.2 of o2](o3){};
  \node[vertex,minimum size=4, right = 1 of o3](o4){};
  \node[vertex,minimum size=4, right = .2 of o4](o5){};
  \node[vertex,minimum size=4, right = 1 of o5](o6){};
  \node[vertex,minimum size=4, right = .2 of o6](o7){};
  \node[vertex,minimum size=4, right = 1 of o7](o8){};  
  \node[vertex,minimum size=4, right = .2 of o8](o9){};
  \node[vertex,minimum size=4, right = 1 of o9](o10){};
  \node[vertex,minimum size=4, right = .2 of o10](o11){};
  \node[vertex,minimum size=4, right = 1 of o11](o12){};
          
        \node[adversary, left = .2 of o1, minimum size=4] (a1) {};  
        \node[adversary,   right = .2 of o12, minimum size=4] (a4) {};  
        
          \foreach \source/ \dest  in {a1/o1, o2/o3, o4/o5, o6/o7, o8/o9, o10/o11,o12/a4}
        \draw[TUMBlau, dotted, ultra thick] (\source) to node{} node[below] {\scriptsize $\varepsilon$} (\dest);
        
          \foreach \source/ \dest  in {o1/o2, o3/o4, o5/o6,o7/o8,o9/o10,o11/o12}
        \draw[TUMBlau, ultra thick] (\source) to node{} node[below] {\scriptsize $1-\varepsilon$} (\dest);
\end{tikzpicture}
%
\caption{Example instance on the line. Vertices in Stage $1$ are depicted by light grey dots and second stage arrivals are indicated as dark grey crosses. First and second state optimum are depicted by solid and dotted edges, respectively. The arrival of two new vertices leads to a significant drop in the cost of the optimum.}\label{fig:motivationKKnown}
\end{figure}
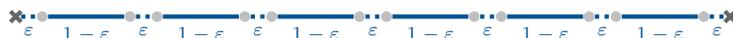

We consider a two-stage robust model with recourse. Assume we are given an underlying metric space $(\mathcal{X},c)$. The input for the \emph{first stage} is a complete graph $G_1$ whose node set $V(G_1)$ is a finite, even subset of $\mathcal{X}$. The cost of an edge $\{v,w\}$ is given by the corresponding cost $c(v,w)$ in the metric space\footnote{Graphs with arbitrary cost functions do not allow for $(O(1),O(1))$-robust matchings in general, e.g., consider a variant of the example in \cref{fig:motivationKKnown} in which all omitted edges have infinite cost.}. In a second stage we get an extended complete graph $G_2$ containing all nodes in $V(G_1)$ plus $2k$ additional nodes. As before, costs of edges in $G_2$ are given by the underlying metric. In the first stage we must create a perfect matching $M_1$ for $G_1$. In the second stage, after $G_2$ is revealed, we must adapt our solution by constructing a new perfect matching $M_2$ for $G_2$, called the \emph{second stage reply}. We say that a solution $M_1$ is two-stage $(\alpha,\beta)$-robust if for any instantiation of the second stage there exists a solution $M_2$ such that two conditions hold. First, the total cost of edges in $M_i$ must satisfy $c(M_i)\le \alpha\cdot c(O_i)$ for $i\in\{1,2\}$, where $O_i$ denotes a min-cost perfect matching in $G_i$. Second, it must hold that $|M_1\setminus M_2|\le \beta k$. An algorithm is two-stage $(\alpha,\beta)$-robust if, given $\gro$ and $c$, it returns a two-stage $(\alpha,\beta)$-robust matching and, given the set of new arrivals, a corresponding second stage reply. We refer to $\alpha$ as the \emph{competitive factor} and $\beta$ as the \emph{recourse factor} of the algorithm. Our main goal is to balance cost and recourse, and thus we aim to obtain algorithms where~$\alpha$ and~$\beta$ are constants.

Our model is closely related to an online model with recourse. Consider a graph whose nodes are revealed online two by two. Our objective is to maintain a perfect matching at all times. As above, irrevocable decisions do not allow for constant competitive factors. This suggests a model where in each iteration we are allowed to modify a constant number of edges. An online algorithm that maintains an $\alpha$-approximation at all time while deleting at most $\beta$ edges per iteration can be easily transformed into a two-stage $(\alpha,\beta)$-robust algorithm. Given an instance of the two-stage model, we choose an arbitrary order for the nodes available in the first stage and create $M_1$ by following the updates proposed by an online algorithm. Then, we repeat the procedure for the arrivals in Stage $2$. Thus, our two-stage model is also the first necessary step for understanding this more involved online model. Megow et al. \cite{2016:Megow} study a similar online model for minimum spanning trees and the TSP in metric graphs. After giving a $(1+\varepsilon)$-competitive algorithm with recourse factor $\frac{1}{\varepsilon} \log (\frac{1}{\varepsilon})$ for the former problem, they are able to derive a $(2+\varepsilon)$-competitive algorithm with constant recourse factor for the latter problem by combining their results with a modified version of the well-known double-tree algorithm. An algorithm for the online variant of our proposed model together with the aforementioned results, would give rise to an online variant of Christofide's algorithm which can yield an improved competitiveness factor for the considered online TSP.

\subparagraph*{Our Results and Techniques}

We distinguish two variants of the model. In the \emph{$k$-known} case we assume that in Stage~1 we already know the number of new nodes $2k$ that will arrive in Stage~2. For this case we present a simple two-stage $(3,1)$-robust algorithm.

\begin{theorem} \label{thm:k-known}
  Let $(\mathcal{X}, c)$ be a metric space, $V_1 \subseteq \mathcal{X}$ with $|V_1|$ even, and $\gro$ be the complete graph on $V_1$. For $k\in \mathbb{N}$ known in advance, there is a perfect matching $\algo$ in $\gro$ that is two-stage $(3, 1)$-robust for $2k$ arrivals. Such a matching and corresponding second stage reply can be computed in time $\text{poly}(|V_1|,k)$.
\end{theorem}

The example in \cref{fig:motivationKKnown} illustrates a worst case scenario for the strategy of choosing $\opto$ as the first stage matching for $k=1$.
The reason for this is that the nodes arriving in Stage~2 induce a path in $\opto \Delta \optt$ that incurs a significant drop in the optimal cost.
Our algorithm is designed towards preparing for such bad scenarios. 
To this end, we define the notion of \emph{gain} for a path $P$ with respect to a matching $X$ as follows:
$$\g[X]{P} := c(P \cap X) - c(P \setminus X).$$
In Stage 1, our algorithm chooses $k$ edge-disjoint $\opto$-alternating paths of maximum total gain with respect to $\opto$. 
For each such path $P$ we modify $\opto$ by removing $P \cap \opto$ and adding $(P\setminus \opto) \cup \{e(P)\}$, where $e(P)$ is the edge that connects the endpoints of $P$.
Our choice of paths of maximum gain implies that $P \cap \opto$ is larger than $P\setminus \opto$. 
Therefore we can bound the cost of the solution in the first stage against that of $\opto$ and also infer that most of its costs is concentrated on the edges $e(P)$. 
For the second stage we construct a solution for the new instance by removing the $k$ edges of the form $e(P)$ and adding new edges on top of the remaining solution. The algorithm is described in detail in \cref{sec:known}.

For the case where $k$ is unknown the situation is considerably more involved as a first stage solution must work for any number of arriving nodes simultaneously. In this setting we restrict our study to the real line and give an algorithm that is two-stage $(10,2)$-robust.

\begin{theorem}\label{thm:k-unknown}
  Let $\mathcal{X} = \mathbb{R}$ and $c = |\cdot|$, $V_1 \subseteq \mathcal{X}$ with $|V_1|$ even, and let $\gro$ be the complete graph on $V_1$.
  Then there is a perfect matching $\algo$ in $\gro$ that is two-stage $(10, 2)$-robust. Such a matching, as well as the second stage reply, can be computed in time $\text{poly}(|V_1|,k)$. 
\end{theorem}

The first stage solution $M$ is constructed iteratively, starting from the optimal solution. We will choose a path $P$ greedily such that it maximizes $\g[M]{P}$ among all alternating paths that are \emph{heavy}, i.e., the cost of $P \cap M$ is a factor 2 more expensive than the cost of $P\setminus M$. Then $M$ is modified by augmenting along $P$ and adding edge $e(P)$, which we fix to be in the final solution. We iterate until $M$ only consists of fixed edges. As we are on the line, each path $P$ corresponds to an interval and we can show that the constructed solution form a laminar family. Furthermore, our choice of heavy paths implies that their lengths satisfy an exponential decay property. This allows us to bound cost of the first stage solution. For the second stage, we observe that the symmetric difference $\opto \Delta \optt$ induces a set of intervals on the line.
For each such an interval, we remove on average at most two edges from the first stage matching and repair the solution with an optimal matching for the exposed vertices.
A careful choice of the removed edges, together with the greedy construction of the first stage solution, enables us to bound the cost of the resulting second stage solution within a constant factor of the optimum.
See \cref{sec:unknown-stage1,sec:unknown-stage2} for a detailed description of this case.
\subparagraph*{Related Work} Intense research has been done on several variants of the \textit{online bipartite matching} problem~\cite{1990:Karp,1993:Kalyanasundaram,1994:Khuller,2009:Chaudhuri,2017:Nayyar}. In this setting we are given a known set of \emph{servers} while a set of \emph{clients} arrive online. In the \textit{online bipartite metric matching problem} servers and clients correspond to points from a metric space. Upon arrival, each client must be matched to a server irrevocably, at cost equal to their distance. For general metric spaces, there is a tight bound of $(2n-1)$ on the competitiveness factor of deterministic online algorithms, where $n$ is the number of servers~\cite{1994:Khuller,1993:Kalyanasundaram}. Recently, Raghvendra presented a deterministic algorithm~\cite{2016:Raghvendra} with the same competitiveness factor, that in addition is $O(\log(n))$-competitive in the random arrival model. Also, its analysis can be 
parameterized for any metric space depending on the length of a TSP tour and its diameter~\cite{2017:Nayyar}. For the special case of the metric on the line, Raghvendra~\cite{2018:Raghvendra} recently refined the analysis of the competitive ratio to $O(\log(n))$. This gives a deterministic algorithm that matches the previously best known bound by Gupta and Lewi~\cite{2012:Gupta}, which was attained by a randomized algorithm. As the lower bound of $9.001$ \cite{2005:Fuchs} could not be improved for 20 years, the question whether there exists a constant competitive algorithm for the line remains open.

The \textit{online matching with recourse problem} considers an unweighted bipartite graph. Upon arrival, a client has to be matched to a server and can be reallocated later. The task is to minimize the number of reallocations under the condition that a maximum matching is always maintained. The problem was introduced by Grove, Kao and Krishnan~\cite{1995:Grove}. Chaudhuri et al.~\cite{2009:Chaudhuri} showed that for the random arrival model a simple greedy algorithm uses $O(n \log(n))$ reallocations with high probability and proved that this analysis is tight. Recently, Bernstein, Holm and Rotenberg \cite{2018:Bernstein} showed that the greedy algorithm needs $O(n \log^2 n)$ allocations in the adversarial model, leaving a small gap to the lower bound of $O(n \log n)$. Gupta, Kumar and Stein~\cite{2014:Gupta} consider a related problem where servers can be matched to more than one client, aiming to minimize the maximum number of clients that are assigned to a server. They achieve a constant competitive factor server while doing in total $O(n)$ reassignments. 

Online min-cost problems with reassignments have been studied in other contexts. For example in the \emph{online Steiner tree problem with recourse} a set of points on a metric space arrive online. We must maintain Steiner trees of low cost by performing at most a constant (amortized) number of edge changes per iteration. While the pure online setting with no reassignment only allows for $\Omega(\log(n))$ competitive factors, just one edge deletion per iteration is enough to obtain a constant competitive algorithm~\cite{2013:Gu}; see also~\cite{2014:Guptab,2016:Megow}.

The concept of \emph{recoverable robustness} is also related to our setting~\cite{2009:Liebchen}. In this context the perfect matching problem on unweighted graphs was considered by Dourado et. al.~\cite{2015:Dourado}. They seek to find perfect matchings which, after the failure of some edges, can be recovered to a perfect matching by making only a small number of modifications. They establish computational hardness results for the question whether a given graph admits a robust recoverable perfect matching.

\section{Known Number of Arrivals}\label{sec:known}

In this section, we consider the setting where $k$ (number of arrival pairs in Stage 2) is already known in Stage 1.
Let $\gro$ be the complete graph given in Stage~1 (with edge costs $c$ induced by an arbitrary metric) and let $\opto$ be a min-cost perfect matching in $\gro$.
Without loss of generality assume that $|\opto| > 2k$, as otherwise, we can remove all edges of $\algo$ in Stage 2.

\begin{description}
\item[\algKnownSto]{ works as follows: 
\begin{romanenumerate}
\item Let $P_1, \dots, P_k$ be edge-disjoint, $\opto$-alternating paths maximizing $\sum_{i=1}^{k}\g[\opto]{P_i}$. 
\item Set $\core := \opto \Delta P_1 \Delta \dots \Delta P_k$.
\item Return $\algo := \core \cup \{e(P_i) : i \in [k]\}.$
\end{romanenumerate}}
\end{description}

It is easy to see that each path $P_i$ starts and ends with an edge from $\opto$ and $\g[\opto]{P_i} \geq 0$. As a consequence, $\algo$ is a perfect matching and
\begin{align*}
\textstyle c(\core) = c(\opto) - \sum_{i=1}^{k} \g[\opto]{P_i} \leq c(\opto).
\end{align*}
Using $c(e(P_i)) \leq c(P_i)$ and $\bigcup_{i=1}^{k} P_i = \opto \Delta \core \subseteq \opto \cup \core$ we obtain
\begin{align*}
\textstyle c(\algo) \leq c(\core) + \sum_{i=1}^{k} c(P_i) \leq c(\core) + c(\core) + c(\opto) \leq 3 \cdot c(\opto).
\end{align*}

Now consider the arrival of $2k$ new vertices, resulting in the graph $\grt$ with min-cost matching $\optt$.
Note that $\optt \Delta \core$ is a $U$-join, where $U$ is the union of the endpoints of the paths $P_1, \dots, P_k$ and the $2k$ newly arrived vertices.

\begin{description}
\item[\algKnownStt] works as follows:
\begin{romanenumerate}
\item Let $P'_1, \dots, P'_{2k}$ be the $2k$ maximal paths from $\optt \Delta \core$.
\item Return $\algt := \core \cup \{e(P'_i) : i \in [2k]\}$.
\end{romanenumerate}
\end{description}

Note that $\opto \Delta \optt$ consists of $k$ alternating paths $R_1, \dots, R_k$, from which we remove the starting and ending $\optt$-edge.
Then these paths would have been a feasible choice for $P_1, \dots, P_k$, implying that the total gain of the $R_i$'s is at most that of the $P_i$'s. We conclude that
\begin{align*}
\textstyle c(\core) =  c(\opto) - \sum_{i=1}^{k} \g[\opto]{P_i} \leq c(\opto) - \sum_{i=1}^{k} \g[\opto]{R_i} \leq c(\optt).
\end{align*}
Applying $\bigcup_{i=1}^{2k} P'_i \subseteq \optt \Delta \core \subseteq \optt \cup \core$, we obtain
\begin{align*}
\textstyle c(\algt) \leq c(\core) + \sum_{i=1}^{2k} c(P'_i) \leq c(\core) + c(\core) + c(\optt) \leq 3 \cdot c(\optt).
\end{align*}

As $|\algo \setminus \algt| \leq |\algo \setminus \core| = k$, we conclude that $\algo$ is indeed two-stage $(3,1)$-robust. 

We remark that $\core$ is always a min-cost matching of cardinality $\vert V(G_1) \vert -2k$ in $G_1$. Thus, alternatively to \algKnownSto, we can compute a min-cost matching of cardinality $\vert V(G_1) \vert-2k$ and match uncovered nodes at minimum cost. Finding $\core$ directly as well as finding gain-maximizing paths as described in \algKnownSto can be done efficiently by solving a min-cost $T$-join problem in an extension of $\gro$. This concludes our proof of \cref{thm:k-known}.


\section{Unknown Number of Arrivals -- Stage 1}\label{sec:unknown-stage1}

In this section, we consider the case that the underlying metric corresponds to the real line. This implies that there is a Hamiltonian path $\LINE$ in $\gro$ such that $c(v, w) = c(\LINE[v, w])$ for all $v, w \in V(\gro)$, where $L[v, w]$ is the subpath of $L$ between nodes $v$ and $w$.
We will refer to $\LINE$ as \emph{the line} and call the subpaths of $\LINE$ \emph{intervals}. 
The restriction to the metric on the line results in a uniquely defined min-cost perfect matching $\opto$ with a special structure. All proofs omitted due to space constraints can be found in the appendix.   

\begin{restatable}{lemma}{restateLemOptOntheLine}
\label{lem:structure-opt-ontheline}
$\opto$ is the unique perfect matching contained in $L$. 
\end{restatable}

When the number of arrivals is not known in the first stage, the approach for constructing the first stage matching introduced in \cref{sec:known} does not suffice anymore. \cref{fig:badExample} illustrates a class of instances for which \algKnownSto cannot achieve $(O(1), O(1))$-robustness, no matter how we choose $k$. 
For a matching $M$, define $g(M) := \max_{e \in L} |\{\{v, w\} \in M : e \in L[v, w]\}|$. Informally speaking, $g(M)$ captures the maximal number of times a part of the line is traversed by edges in $M$. 
The example in \cref{fig:badExample} can be generalized to show that we cannot restrict ourselves to constructing matchings $\algo$ such that $g(\algo)$ is bounded by a constant.
\begin{figure}[t]
\centering
\begin{tikzpicture}[scale=.25]
    \foreach \pos/\name/\lab in { {(0,1)/v0/}, {(8,1)/v10/}, {(10,1)/v11/}, {(11,1)/v115/}, {(14,1)/v145/}, {(15,1)/v15/}, {(17,1)/v16/}, {(25,1)/v26/}, {(27,1)/v27/}, {(28,1)/v275/}, {(31,1)/v305/},{(32,1)/v31/},{(34,1)/v32/}, {(42,1)/v42/}} 
        \node[vertex,minimum size=5, label=below:{\lab}](\name) at \pos {};

    \foreach \pos/\name/\lab in { {(46,1)/v44/}, {(46.5,1)/v46/}}
        \node[vertex,minimum size=0, label=below:{\lab}](\name) at \pos {};

    \foreach \source/ \dest/\lab  in {v0/v10/$c_1$,v11/v115/$c_3$, v145/v15/$c_3$,v16/v26/$c_1$, v27/v275/,v305/v31/,v32/v42/$c_1$} 
     \path[draw, TUMBlau, ultra thick, below] (\source) edge node[] {\textcolor{black}{\small \lab}} (\dest);
     
    \foreach \source/ \dest/\lab  in {v10/v11/$c_2$,v15/v16/$c_2$}
     \path[draw=black, opacity=0, below] (\source) edge node[color=black, opacity=1] {\textcolor{black}{\small \lab}} (\dest);
     
     \path (v115) -- node{\tiny [\dots]} (v145);        
     \path (v275) -- node{\tiny [\dots]} (v305);  
     \path (v42) -- node{[\dots]} (v44);                      

\end{tikzpicture}
\caption{For fixed $\alpha, \beta \in O(1)$, we construct an instance for which \algKnownSto is not $(\alpha,\beta)$-robust for any $k \in O(1)$. $\gro$ is constructed such that $\opto$ contains $\beta + 1$ edges of size $c_1$ and $\beta^3 + \beta^2$ of size $c_3$ that are equally distributed between $c_1$-edges. The distance between any two consecutive edges in $\opto$ is $c_2$. Values $c_1,c_2,c_3$ are chosen depending on $\alpha$, guaranteeing $c_1 \gg \beta^3 c_2 \gg \beta^6 c_3$. \algKnownSto with $k \geq \beta + 1$ chooses $M_1 = O_1$. Now, assume there are two arrivals at the extremes of the line: while the optimal costs in the second stage decrease heavily to $c(O_2) \in \Theta(\beta^3 c_2)$, only $\beta$ deletions are allowed within $M_1$. As a result, there does not exist a feasible second stage reply. Now, consider \algKnownSto with $k \leq \beta$. Then, $M_1$ contains more than $\beta^2 + \beta$ edges of size $c_2$. If in the second stage $2(\beta + 1)$ nodes arrive right next to the endpoints of $c_1$-edges, the optimal costs drop to $ \Theta(\beta^3 c_3)$ while only $\beta^2 + \beta$ deletions are allowed. Again, no feasible second stage reply exists.}\label{fig:badExample}
\end{figure}
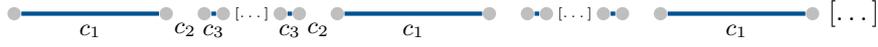

In view of the example in \cref{fig:badExample}, we adapt the approach from \cref{sec:known} as follows.
Instead of creating a fixed number of paths, our algorithm now iteratively and greedily selects a path $P$ of maximum gain with respect to a dynamically changing matching $X$ (initially $X = \opto$). 
In order to bound the total cost incurred by adding edges of the form $e(P)$, we only consider paths $P$ for which $X \cap P$ contributes a significant part to the total cost of $P$.

\begin{definition}
Let $X, P \subseteq E(\gro)$. 
\begin{enumerate}
\item We say that $P$ is \textit{$X$-heavy} if $c(P \cap X ) \geq 2 \cdot c(P \setminus X).$
\item We say that $P$ is \textit{$X$-light} if $c(P \cap X ) \leq \frac{1}{2} \cdot c(P \setminus X).$
\end{enumerate}
\end{definition}

\begin{description}
\item[\algUnknownSto]{ works as follows: 
\begin{romanenumerate}
\item Set $\algo := \emptyset$ and $X~:=~\opto$.
\item While $X \neq \emptyset$: Let $P$ be an $X$-heavy $X$-alternating path maximizing $\g[X]{P}$ and update $\algo \leftarrow \algo \cup \{e(P)\}$ and  $X \leftarrow X \Delta P$.
\item Return $\algo$.
\end{romanenumerate}
} 
\end{description}

Note that in each iteration, the path $P$ starts and ends with an edge from $X$ as it is gain-maximizing (if $P$ ended with an edge that is not in $X$, we could simply remove that edge and obtain a path of higher gain). Therefore it is easy to see that $X \cup \algo$ is always a perfect matching, and in each iteration the cardinality of $X$ decreases by $1$.

Now number the iterations of the while loop in \algUnknownSto from $1$ to $n$. Let $X^{(i)}$ be the state of $X$ at the beginning of iteration $i$. Let $P^{(i)}$ be the path chosen in iteration $i$ and let $e^{(i)} = e(P^{(i)})$ be the corresponding edge added to $M$. 
The central result in this section is \cref{lem:laminarity},  in which we show that the paths $P^{(i)}$ form a laminar family of intervals on the line.

Within the proof we will make use of observations stated in \cref{lem:heavy-union,lem:gain-maximizers}. For convenience, we define the projection $\PL{e} := \LINE[v, w]$ that maps an edge $e = \{v, w\} \in E(\gro)$ to the corresponding subpath $\LINE[v, w]$. $\;$\\

\begin{restatable}{lemma}{restateLemHeavyUnion}\label{lem:heavy-union}\label{lem:heavy-setminus}
  Let $X \subseteq E(\gro)$.
  \begin{enumerate}
  \item Let $A, B \subseteq E(\gro)$ be two $X$-heavy ($X$-light, respectively) sets with $A \cap B = \emptyset$. Then $A \cup B$ is $X$-heavy ($X$-light, respectively).
  \item Let $A, B \subseteq E(\gro)$ with $B \subseteq A$. If $A$ is $X$-heavy and $\g[X]{B} < 0$, then $A \setminus B$ is $X$-heavy. If $A$ is $X$-light and $\g[X]{B} > 0$, then $A \setminus B$ is $X$-light.
  \end{enumerate}
\end{restatable}

\begin{restatable}{lemma}{restateLemGainMaximizers}\label{lem:gain-maximizers}\label{lem:interval-to-path}
Let $X \subseteq \LINE$ be a matching. 
\begin{enumerate}
\item Let $P$ be an $X$-heavy $X$-alternating path maximizing $\g[X]{P}$. If $X$ covers all vertices in $V(\PL{P})$, then \mbox{$P=\PL{P}$}.
\item Let $I \subseteq L$ be an interval. Then there is an $X$-alternating path $P$ such that $c(P \cap X) = c(X \cap I)$ and $c(P \setminus X) = c(I \setminus X)$. 
\end{enumerate}
\end{restatable}

  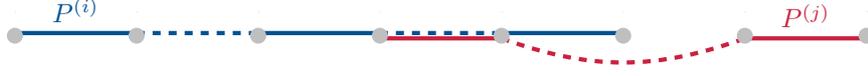
\begin{figure}[t]
\centering
      \begin{tikzpicture}[scale=.4]
      \foreach \pos/\name/\lab in { 
      {(0,1)/v0}, 
    {(4,1)/v1}, 
    {(8,1)/v2}, 
    {(12,1)/v3}, 
    {(16,1)/v4}, 
    {(20,1)/v5},
    {(24,1)/v6},
    {(28,1)/v7}}
         \node[vertex,minimum size=0, color=black!30](\name) at \pos {};
               \foreach \source/ \dest  in {{v0/v1}}
        \draw [TUMBlau, ultra thick,above] ($(\source) - (0,2em)$) to node{$P^{(i)}$} ($(\dest) - (0,2em)$); 
                       \foreach \source/ \dest  in { {v2/v3},{v4/v5}}
        \draw [TUMBlau, ultra thick,above] ($(\source) - (0,2em)$) to node{} ($(\dest) - (0,2em)$); 
        
                       \foreach \source/ \dest  in {{v1/v2},{v3/v4}}
        \draw [TUMBlau, ultra thick, dashed,above] ($(\source) - (0,2em)$) to node{} ($(\dest) - (0,2em)$); 
               \foreach \source/ \dest  in {{v3/v4}}
        \draw [TUMRot, ultra thick, below] ($(\source) - (0,2.5em)$) to node {} ($(\dest) - (0,2.5em)$); 
                       \foreach \source/ \dest  in {{v4/v6}}
        \draw [TUMRot, ultra thick, dashed, above, bend right=20] ($(\source) - (0,2.5em)$) to node {} ($(\dest) - (0,2.5em)$); 
                       \foreach \source/ \dest  in {{v6/v7}}
        \draw [TUMRot, ultra thick, above] ($(\source) - (0,2.5em)$) to node {$P^{(j)}$} ($(\dest) - (0,2.5em)$); 

          \foreach \nod  in {v0,v1,v2,v3,v4,v5,v6,v7}
        \path (\nod) + (0,-2.25em) node[vertex,minimum size=6] {}; 

    \end{tikzpicture}
\caption{A minimal example of a situation in which \cref{lem:laminarity} would be violated. There are two iterations $i<j$ with paths $P^{(i)}$ (depicted in blue) and $P^{(j)}$ (depicted in red) such that $X \cap P^{(j)}$ was not modified between iteration $i$ and iteration $j$. Then, extending the blue path $P^{(i)}$ with the rightmost edge yields an $X^{(i)}$-heavy path with higher gain than $P^{(i)}$, a contradiction.} \label{fig:laminarity-intuition}
\end{figure}

\begin{restatable}{lemma}{restateLemLaminarity}
\label{lem:laminarity}
  \begin{enumerate}
    \item $X^{(i)}, P^{(i)} \subseteq \LINE$ for all $i \in [n]$. \label{lem:laminarity-part1}
    \item For all $i, j \in [n]$ with $i < j$, either $P^{(i)} \cap P^{(j)} = \emptyset$ or $P^{(j)} \subset P^{(i)}$.\label{lem:laminarity:laminarity} \label{lem:laminarity-part2}
  \end{enumerate}
\end{restatable}

\begin{proof}
We say a pair $(i, j)$ with $i < j$ is \emph{violating} if $\PL{P^{(i)}} \cap \PL{P^{(j)}} \neq \emptyset$ and $\PL{P^{(j)}} \setminus \PL{P^{(i)}} \neq \emptyset$. We will show that no violating pair exists. This proves the lemma as the following claim asserts.
  
  \begin{claim*}
    If $P^{(j)} \neq \PL{P^{(j)}}$, then there is a violating pair $(i', j')$ with $i' < j' \leq j$.
  \end{claim*}
  \begin{claimproof}
  Let $j'$ be minimal with $P^{(j')} \neq \PL{P^{(j')}}$.
  Note that minimality of $j'$ implies that $X^{(j')} = \opto \Delta P^{(1)} \Delta \dots \Delta P^{(j'-1)} \subseteq \LINE$. 
  Then \cref{lem:gain-maximizers} implies that there must be a vertex $v \in V(\PL{P^{(j')}})$ not covered by $X^{(j')}$. 
  Because $X^{(j')}$ covers exactly those vertices not covered by $\{e^{(i')} : i' < j'\}$, there must be an $i' < j'$ such that $v$ is an endpoint of $P^{(i')}$.
  The vertex $v$ cannot be an endpoint of $P^{(j')}$, because $v$ is exposed in $X^{(j')}$ and $P^{(j')}$ starts and ends with edges from $X^{(j')}$.
  This implies that $(i', j')$ is a violating pair. 
  \end{claimproof}
  
  Now let us assume there are no violating pairs. Then $P^{(i)} = \PL{P^{(i)}} \subseteq \LINE$ for all $i \in [n]$ by the claim, which also implies $X^{(i)} \subseteq \LINE$. This implies the lemma as, in this situation, the condition for violating pairs coincides with the condition in point \ref{lem:laminarity:laminarity} of the lemma.
  
  By contradiction assume there is a violating pair. Choose $j$ such that $j$ is minimal among all possible choices of violating pairs. 
  Then choose $i$ such that it is maximal for that $j$ among all violating pairs.
  
  Note that the claim implies that $P^{(i)}, X^{(i)}, X^{(j)} \subseteq \LINE$.
  Furthermore, our choice of $i$ and $j$ implies that $\PL{P^{(j')}} \cap \PL{P^{(j)}} = \emptyset$ for all $j'$ with $i < j' < j$, as otherwise $(i, j')$ or $(j', j)$ would be a violating pair.
  In particular, $P^{(j')} \cap P^{(j)} = \emptyset$ for all $j'$ with $i < j' < j$ and thus
  \begin{equation}
    X^{(j)} \cap P^{(j)} = X^{(i + 1)} \cap P^{(j)} = (X^{(i)} \Delta P^{(i)}) \cap P^{(j)}.\label{eq:X-inverted}
  \end{equation}

  Now consider $I_1 := \PL{P^{(j)}} \cap P^{(i)}$ and $I_2 := \PL{P^{(j)}} \setminus P^{(i)}$, both of which are non-empty since $(i, j)$ is a violating pair.
  Then \eqref{eq:X-inverted} implies $X^{(j)} \cap I_1 = I_1 \setminus X^{(i)}$ and $I_1 \setminus X^{(j)} = X^{(i)} \cap I_1$.
 We conclude that $\g[X^{(j)}]{I_1} = -\g[X^{(i)}]{I_1} \leq 0$, as $I_1$ is a prefix of the gain-maximizing path $P^{(i)}$ (see \cref{lem:prefix} in the appendix for a formal proof).
 
  Therefore $I_2 = \PL{P^{(j)}} \setminus I_1$ is $X^{(j)}$-heavy by \cref{lem:heavy-setminus}.
  But then $I_2$ is also $X^{(i)}$-heavy because \eqref{eq:X-inverted} implies $X^{(j)} \cap I_2 = X^{(i)} \cap I_2$.
  Hence $I' := P^{(i)} \cup I_2$ is $X^{(i)}$-heavy by \cref{lem:heavy-union} and further
  \begin{align*}
  \textstyle \g[X^{(i)}]{I'} = \g[X^{(i)}]{P^{(i)}} + \g[X^{(j)}]{I_2} > \g[X^{(i)}]{P^{(i)}},
  \end{align*}
  because $\g[X^{(j)}]{I_2} \geq \frac{1}{3} c(I_2)$.
  By \cref{lem:interval-to-path} there is an $X^{(i)}$-heavy, $X^{(i)}$-alternating path with higher gain than $P^{(i)}$, a contradiction. 
  \end{proof}
  
  \subparagraph*{Tree structure} \cref{lem:laminarity} induces a tree structure on the paths selected by \algUnknownSto.
  We define the directed tree $T=(W,A)$ as follows. We let $W := \{0, \dots n\}$ and define $P^{(0)} := L$.
   For $i, j \in W$ we add the arc $(i, j)$ to $A$ if $P^{(j)} \subset P^{(i)}$ and there is no $i' \in W$ with $P^{(j)} \subset P^{(i')} \subset P^{(i)}$.
  It is easy to see that $T$ is an out-tree with root $0$. We let $T[i]$ be the unique $0$-$i$-path in $T$.
  We define the set of children of $i \in W$ by $\child{i} := \{j \in W : (i, j) \in A\}$.
  Furthermore, let $\HEAVY := \{i \in W : |T[i]| \text{ is odd}\}$ and $\LIGHT := \{i \in W : |T[i]| \text{ is even}\}$ be the set of \emph{heavy} and \emph{light} nodes in the tree, respectively. These names are justified by the following lemma. See \cref{fig:tree} a)-b) for an illustration. 
  
  \begin{figure}[t!]
\begin{center}
\begin{minipage}[b]{0.45\linewidth}
  \begin{tikzpicture}[scale=.5]

    \foreach \x in {1,...,10}
       {\pgfmathtruncatemacro{\label}{\x +1}
       \node [vertex,minimum size=4]  (\x) at (\x,9) {};} 
       
        \foreach \source/ \dest  in {3/10}
        \draw[TUMBlau, bend left=50, ultra thick, above] (\source) to node{\tiny{$1$}} (\dest);
                \foreach \source/ \dest  in {4/7}
        \draw[TUMRot, dotted, bend left=50, ultra thick, above] (\source) to node{\tiny{$2$}} (\dest);
        
                        \foreach \source/ \dest  in {8/9}
        \draw[TUMRot,dotted, ultra thick, above] (\source) to node{\tiny{$5$}} (\dest);
        
                \foreach \source/ \dest  in {1/2}
        \draw[TUMBlau, ultra thick, above] (\source) to node{\tiny{$3$}} (\dest);
        
                        \foreach \source/ \dest  in {5/6}
        \draw[TUMBlau, ultra thick,above] (\source) to node{\tiny{$4$}} (\dest);
   \node (d1) at (0,10) {a)};
   
       \foreach \x in {1,...,10}
       {\pgfmathtruncatemacro{\label}{\x +1}
       \node [vertex,minimum size=4]  (\x) at (\x,5) {};} 
       
        \foreach \source/ \dest  in {3/10}
        \draw[TUMBlau, bend left=50, ultra thick] (\source) to node{} (\dest);
                \foreach \source/ \dest  in {4/7}
        \draw[TUMRot, dotted, bend left=50, ultra thick, above] (\source) to node{\tiny{$2$}} (\dest);
        
                        \foreach \source/ \dest  in {8/9}
        \draw[TUMRot,dotted, ultra thick, above] (\source) to node{\tiny{$5$}} (\dest);
        
                \foreach \source/ \dest  in {5/6, 1/2}
        \draw[TUMBlau, ultra thick] (\source) to node{} (\dest);
        
        \draw [black!30, thick,below] (1) -- ++(0,-1em) -- node[below] {\footnotesize{$R$}} ($(4) - (0,1em)$) -- (4); 
         \draw [black!30, thick,below] (5) -- ++(0,-1em) -- node[below] {$R'$} ($(6) - (0,1em)$) -- (6); 
          \draw [black!30, thick,below] (7) -- ++(0,-1em) -- node[below] {$R''$} ($(10) - (0,1em)$) -- (10); 
          
             \draw [black!70, dashed, thick,below] (4) -- ++(0,-4em) -- node[below] {$\bar{R}$} ($(7) - (0,4em)$) -- (7); 
        
   \node (d1) at (0,6) {c)};

  \end{tikzpicture}
  \end{minipage}
  \begin{minipage}[b]{0.54\linewidth}
   \begin{tikzpicture}[scale=.45]

\tikzset{my node/.style={node distance=0.1cm,shape=rectangle}}
    \node[draw,color=black,  rectangle, rounded corners=1ex,
      align=center, inner sep=1ex] (root) at (2,0){
      \begin{minipage}{0.38\textwidth}
      \centering
      \begin{tikzpicture}[scale=.3]
     \foreach \x in {1,...,10}
       {\pgfmathtruncatemacro{\label}{\x +1}
       \node [vertex,minimum size=4]  (\x) at (\x,1) {};} 
       
        \foreach \source/ \dest  in {1/2,3/4,5/6,7/8,9/10}
        \draw[black!70, ultra thick] (\source) to node{} (\dest);
    \end{tikzpicture}
\end{minipage}};
  \path (-3, -2.5)  node[draw,color=black,   rectangle, rounded corners=1ex,
      align=center, inner sep=1ex] (one) {
      \begin{minipage}{0.1\textwidth}
      \centering
      \begin{tikzpicture}[scale=.5]
       \foreach \x in {1,...,2}
       {\pgfmathtruncatemacro{\label}{\x +1}
       \node [vertex,minimum size=4]  (\x) at (\x,1) {};} 
       
        \foreach \source/ \dest  in {1/2}
        \draw[black!70, ultra thick] (\source) to node{} (\dest);
        
         \foreach \source/ \dest  in {1/2}
        \draw[TUMBlau, bend left = 50,  ultra thick] (\source) to node{} (\dest);

    \end{tikzpicture} 
    \end{minipage}
    };
        \path (5, -2.5) node[draw,color=black,    rectangle, rounded corners=2ex,
      align=center, inner sep=0.7ex] (two) {
      \begin{minipage}{0.28\textwidth}
      \centering
      \begin{tikzpicture}[scale=.3]
       \foreach \x in {1,...,8}
       {\pgfmathtruncatemacro{\label}{\x +1}
       \node [vertex,minimum size=4]  (\x) at (\x,1) {};} 
       
        \foreach \source/ \dest  in {1/2,3/4,5/6,7/8}
        \draw[black!70, ultra thick] (\source) to node{} (\dest);
                 \foreach \source/ \dest  in {1/8}
        \draw[TUMBlau, bend left=30,  ultra thick] (\source) to node{} (\dest);
        
    \end{tikzpicture} 
    \end{minipage}
    };
     \path (2, -5) node[draw,color=black,  rectangle, rounded corners=2ex,
      align=center, inner sep=1ex] (three) {
      \begin{minipage}{0.2\textwidth}
      \centering
      \begin{tikzpicture}[scale=.5]
       \foreach \x in {1,...,4}
       {\pgfmathtruncatemacro{\label}{\x +1}
       \node [vertex,minimum size=4]  (\x) at (\x,1) {};} 
       
        \foreach \source/ \dest  in {1/2,3/4}
        \draw[black!70, ultra thick] (\source) to node{} (\dest);    
                 \foreach \source/ \dest  in {1/4}
        \draw[TUMRot, dotted, bend left= 40, ultra thick] (\source) to node{} (\dest);
        \end{tikzpicture} 
    \end{minipage}
    };
     \path (8, -5) node[draw,color=black,   rectangle, rounded corners=2ex, align=center, inner sep=1ex] (four) {
           \begin{minipage}{0.1\textwidth}
           \centering
      \begin{tikzpicture}[scale=.5]
       \foreach \x in {1,...,2}
       {\pgfmathtruncatemacro{\label}{\x +1}
       \node [vertex,minimum size=4]  (\x) at (\x,1) {};} 
       
        \foreach \source/ \dest  in {1/2}
        \draw[black!70, ultra thick] (\source) to node{} (\dest);
                 \foreach \source/ \dest  in {1/2}
        \draw[TUMRot, dotted, bend left = 50, ultra thick] (\source) to node{} (\dest);
    \end{tikzpicture} 
    \end{minipage}
    };
    \path (-1, -7.5) node[draw,color=black, rectangle, rounded corners=2ex, align=center, inner sep=1ex] (five) {
           \begin{minipage}{0.1\textwidth}
           \centering
      \begin{tikzpicture}[scale=.5]
       \foreach \x in {1,...,2}
       {\pgfmathtruncatemacro{\label}{\x +1}
       \node [vertex,minimum size=4 ]  (\x) at (\x,1) {};} 
       
        \foreach \source/ \dest  in {1/2}
        \draw[black!70, ultra thick] (\source) to node{} (\dest);
                 \foreach \source/ \dest  in {1/2}
        \draw[TUMBlau, bend left = 50, ultra thick] (\source) to node{} (\dest);
    \end{tikzpicture} 
    \end{minipage}
    };    
    
   \node (d0) at (-5,1) {b)};
    
    \node[scale=0.9] (d1) at ($(root.north west) + (-.25,.25)$) {0};
   \node[scale=0.9]  (d2) at ($(one.north west) + (-.25,.25)$) {{3}};
   \node[scale=0.9]  (d3) at ($(two.north west) + (-.25,.25)$) {{1}};
     \node[scale=0.9]  (d4) at ($(three.north west) + (-.25,.25)$) {{2}};
   \node[scale=0.9]  (d5) at ($(four.north west) + (-.5,0)$) {{5}};
   \node[scale=0.9]  (d6) at ($(five.north west) + (-.25,.25)$) {{4}};
       \draw[ black, -> ] (root) -- (one);
    \draw[ black,-> ] (root) -- (two);
    \draw[ black , ->] (two) -- (three);
    \draw[ black , ->] (two) -- (four);
    \draw[ black , ->] (three) -- (five);
  \end{tikzpicture}
  \end{minipage}
  \vspace{-1cm}
  \end{center}
  \caption{a) Illustration of the matching created by an example execution of \algUnknownSto. Edges added to $M_1$ in an iteration from $\HEAVY$ are depicted by blue solid lines and edges created in an iteration from $\LIGHT$ are illustrated by red dotted lines. b) Illustration of the corresponding tree. For every tree-node $i \in W$, grey edges indicate $X^{(i)}$ and an arc illustrates the edge connecting the end nodes of $P^{(i)}$. c) Illustration of example assignment of requests to iterations (defined in \cref{sec:unknown-stage2}). Requests $R,R',R'' \in \mathcal{R}$ are assigned such that $R,R'' \in \mathcal{R}(0)$ and $R' \in \mathcal{R}(2)$. $\bar{R} \in \mathcal{\bar{R}}(0)$ is a gap between two requests associated with tree-node $0$.} 
  \label{fig:tree}
\vspace{-.5cm}
  \end{figure}
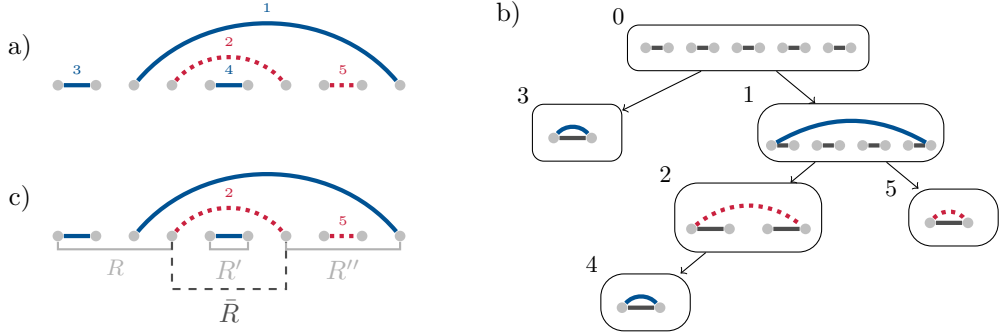
    
  \begin{restatable}{lemma}{restateLemHeavyLight}\label{lem:heavy-light}
    If $i \in \HEAVY$, then $P^{(i)} \cap X^{(i)} = P^{(i)} \cap \opto$ and, in particular, $P^{(i)}$ is $\opto$-heavy.
    If $i \in \LIGHT \setminus \{0\}$, then $P^{(i)} \cap X^{(i)} = P^{(i)} \setminus \opto$ and, in particular, $P^{(i)}$ is $\opto$-light.
  \end{restatable}
  
    \begin{proof}
Let $i \in W \setminus \{0\}$. From \cref{lem:laminarity} we know that for every iteration $i' \in W \setminus \{0\}$ with $i' < i$ it holds that either $P^{(i)} \cap P^{(i')} = \emptyset$ or $P^{(i)} \subset P^{(i')}$. In the first case it holds that $P^{(i)} \cap X^{(i'+1)} = P^{(i)} \cap X^{(i')}$, in the latter case it holds that $P^{(i)} \cap X^{(i'+1)} = P^{(i)} \cap (P^{(i)} \Delta X^{(i')})$. Moreover, it is easy to see that $i'<i$ and $P^{(i)} \subset P^{(i')}$ holds if and only if $i' \in V(T[i])\setminus \{0,i\}$. If $i \in \HEAVY$, this implies that there exist an even number of iterations $i'<i$ for which $P^{(i)} \subset P^{(i')}$ holds. Hence, we obtain 
\[P^{(i)} \cap X^{(i)} = P^{(i)} \cap (\underbrace{P^{(i)} \Delta \dots \Delta P^{(i)}}_{\text{evenly often}} \Delta \, \opto) = P^{(i)} \cap \opto.\] 
If $i \in \LIGHT$, this implies that there exist an odd number of iterations $i'<i$ for which $P^{(i)} \subset P^{(i')}$ holds. Hence, we can deduce that 
\[P^{(i)} \cap X^{(i)} = P^{(i)} \cap (\underbrace{P^{(i)} \Delta \dots \Delta P^{(i)}}_{\text{oddly often}} \Delta \, \opto) = P^{(i)} \setminus \opto.\qedhere\]\end{proof}
  
  The fact that nested paths are alternatingly $\opto$-heavy and $\opto$-light implies an exponential decay property. As a consequence we can bound the cost of $M_1$. 
  
  \begin{restatable}{lemma}{restateLemExpDecay} \label{lem:exp-decay}
    Let $i \in W \setminus \{0\}$. Then
    $\sum_{j \in \child{i}} c(P^{(j)}) \leq \frac{1}{2} \cdot c(P^{(i)})$.
  \end{restatable}
  
    \begin{proof}
  Let $i \in W \setminus \{0\}$.
  Then
  \begin{align*}
    \textstyle
    \sum_{j \in \child{i}} \!\! c(P^{(j)}) \leq \frac{3}{2} \sum_{j \in \child{i}} c(P^{(j)} \cap X^{(j)}) \leq \frac{3}{2} c(P^{(i)} \setminus X^{(i)}) \leq \frac{1}{2} c(P^{(i)}),
  \end{align*}
  where the first inequality follows from the fact that $P^{(j)}$ is $X^{(j)}$-heavy;
  the second inequality follows from the fact that $P^{(j)} \cap X^{(j)} \subseteq P^{(i)} \setminus X^{(i)}$ for $j \in \child{i}$ and the fact that the intervals $P^{(j)}$ for all children are disjoint; 
  the last inequality follows from the fact that $P^{(i)}$ is $X^{(i)}$-heavy. 
  \end{proof}
 
  \begin{restatable}{lemma}{restateBoundedCostUnknown} 
    $c(\algo) \leq 3 c(\opto)$.
  \end{restatable}
    
  \begin{proof}
    Note that $c(\algo) = \sum_{i=1}^{n} c(e^{(i)}) = \sum_{i \in W \setminus \{0\}} c(P^{(i)})$.
    For $\ell \in \mathbb{N}$, let \begin{align*} \textstyle W_{\ell} := \{i \in W : |T[i]| = \ell\}.\end{align*}
    Observe that \cref{lem:exp-decay} implies that
     $\sum_{i \in W_{\ell}} c(P^{(i)}) \leq \left(\frac{1}{2}\right)^{\ell - 1} \sum_{i \in W_1} c(P^{(i)})$
    for all $\ell \in \mathbb{N}$.
    Furthermore
    $\sum_{i \in W_1} c(P^{(i)}) \leq \frac{3}{2} c(\opto)$,
    because $W_1 \subseteq \HEAVY$. Hence
    \[c(M_1) = \sum_{\ell = 1}^{\infty} \sum_{i \in W_\ell} c(P^{(i)}) \leq \sum_{\ell = 1}^{\infty} \left(\frac{1}{2}\right)^{\ell - 1} \sum_{i \in W_1} c(P^{(i)}) = 2 \cdot \frac{3}{2} c(\opto).\qedhere\]
  \end{proof}

\section{Unknown Number of Arrivals -- Stage 2}\label{sec:unknown-stage2}

We now discuss how to react to the arrival of $2k$ additional vertices.
We let $\optt$ be the min-cost perfect matching in the resulting graph $\grt$ and define
\begin{align*} \textstyle \mathcal{R} := \{P : P \text{ is a maximal path in } (\opto \Delta \, \optt) \cap \LINE\}. \end{align*}
We call the elements of $\mathcal{R}$ \emph{requests}. An important consequence of our restriction to the metric space on the line is that $|\mathcal{R}| \leq k$ (in fact, each of the $k$ maximal paths of $\opto \Delta \optt$ is contained in $L$ after removing its first and last edge).

\begin{restatable}{lemma}{restateLemRequestStructure}
\label{lem:request-structure}
  $|\mathcal{R}| \leq k$ and each $R \in \mathcal{R}$ starts and ends with an edge of $\opto$.
\end{restatable}

For simplification of the analysis we make the following assumptions on the structure of the request set. 

\begin{assumption}\label{ass:request-light-laminarity}
  For all $i \in \LIGHT$ and all $R \in \mathcal{R}$, either $P^{(i)} \cap R = \emptyset$ or $P^{(i)} \subseteq R$, or $R \subseteq P^{(i)}$.
\end{assumption}

\begin{assumption}\label{ass:request-heavy-prefixes}
  For all $j \in \HEAVY$, if $\bigcup_{R \in \mathcal{R}}R \cap P^{(j)} \neq \emptyset$, then the first and last edge of $P^{(j)}$ are in $\bigcup_{R \in \mathcal{R}}R$.
\end{assumption}

In \cref{apx:assumptions-wlog} we prove formally that these are without loss of generality. For intuition, we give a short sketch of the proof. Assume we are confronted with a set of requests $\mathcal{R}$ that violates at least one of the assumptions. Based on $\mathcal{R}$, we can construct a modified set of requests $\mathcal{R'}$ complying with the assumptions and fulfilling two properties: First, $\g[O_1]{\bigcup_{R \in \mathcal{R'}}R} > \g[O_1]{\bigcup_{R \in \mathcal{R}}R}$, i.e., $R'$ induces smaller second stage optimal costs than $\mathcal{R}$ and secondly, $\vert \mathcal{R'} \vert \leq \vert \mathcal{R} \vert$, i.e., $\mathcal{R'}$ allows for at most as many modifications as $\mathcal{R}$ does. As a consequence, if we run our proposed second stage algorithm for $\mathcal{R'}$ and construct $M_2$ accordingly, the analysis of the approximation and recourse factor carry over from the analysis of $\mathcal{R'}$ to the actual set of requests $\mathcal{R}$. Hence, we can assume w.l.o.g. that $\mathcal{R}$ fulfills the assumptions.

From the set of requests $\mathcal{R}$, we will determine a subset of at most $2k$ edges that we delete from $\algo$.
To this end, we assign each request to a light node in $\LIGHT$ as follows.
For $R \in \mathcal{R}$ we define $i_R := \max \{i \in \LIGHT : R \subseteq P^{(i)} \}$, i.e., $P^{(i_R)}$ is the inclusionwise minimal interval of a light node containing $R$.
For~$i \in \LIGHT$, let 
\begin{align*}
\textstyle \mathcal{R}(i) := \big\{R \in \mathcal{R} : i_R = i \big\}.
\end{align*}

Furthermore, we also keep track of the gaps between the requests in $\mathcal{R}(i)$ as follows. 
For $i \in \LIGHT$, let
\begin{align*}
  \bar{\mathcal{R}}(i) := \big\{\bar{R} \subseteq P^{(i)} : \ & \textstyle \bar{R} \text{ is a maximal path in } P^{(i)} \setminus \bigcup_{R \in \mathcal{R}(i)} R \text{ and } \\
  & \bar{R} \subseteq P^{(j)} \text{ for some } j \in \child{i} \big\}.
\end{align*}
Note that $R' \cap R'' = \emptyset$ for all $R', R'' \in \mathcal{R}(i) \cup \bar{\mathcal{R}}(i), i \in \LIGHT$.
However, $\bar{R} \in \bar{\mathcal{R}}(i)$ may contain a request $R \in \mathcal{R}(j)$ from descendants $j$ of $i$. See \cref{fig:tree} c) for an illustration of the assignment.

For $i \in \LIGHT$, let $\HEAVY(i) := \child{i}$ and $\LIGHT(i) := \{i' \in W : i' \in \child{j} \text{ for some } j \in \HEAVY(i)\}$. Note that $\HEAVY(i) \subseteq \HEAVY$ and $\LIGHT(i) \subseteq \LIGHT$.
Before we can state the algorithm for computing the second stage reply, we need one final lemma.

\begin{restatable}{lemma}{restateLemIandJexist}\label{lem:i-and-j-exist}
  Let $i \in \LIGHT$. For every $R \in \mathcal{R}(i)$, there is a $j \in \HEAVY(i)$ with $P^{(j)} \cap R \neq \emptyset$.
  For every $\bar{R} \in \bar{\mathcal{R}}(i)$, there is an $i' \in \LIGHT(i)$ with $P^{(i')} \cap R \neq \emptyset$.
\end{restatable}

We are now ready to state the algorithm. We first describe and discuss a simplified version, which yields an approximation guarantee of $19$. At the end of the paper, we discuss how to slightly adapt the algorithm so as to obtain the factor of $10$ given in \cref{thm:k-unknown}.

\begin{description}
\item[\algUnknownStt]{ works as follows:
\begin{romanenumerate}
\item{Create the matching $M'$ by removing the following edges from $\algo$ for each $i \in \LIGHT$:
\begin{enumerate}[1.]
  \item The edge $e^{(i)}$ if $i \neq 0$ and $\mathcal{R}(i) \neq \emptyset$.
  \item For each $R \in \mathcal{R}(i)$ the edge $e^{(j^*_R)}$ where $j^*_R := \min \{j \in \HEAVY(i) : P^{(j)} \cap R \neq \emptyset\}$.
  \item For each $\bar{R} \in \bar{\mathcal{R}}(i)$ the edge $e^{(i^*_{\bar{R}})}$ where $i^*_{\bar{R}} := \min \{i' \in \LIGHT(i) : P^{(i')} \cap R \neq \emptyset\}$.
\end{enumerate}}
\item Let $M''$ be a min-cost matching on all vertices not covered by $M'$ in $\grt$.
\item Return $\algt := M' \cup M''$.
\end{romanenumerate}}
\end{description}

Let $Z$ be indices of the edges removed in step (i). It is not hard to see that $|\bar{\mathcal{R}}(i)| \leq |\mathcal{R}(i)| - 1$ for each $i \in \LIGHT$ and therefore $|Z| \leq 2k$, bounding the recourse of {\algUnknownStt} as intended.

\begin{restatable}{lemma}{restateLemRecourseBoundedUnknown}
\label{lem:recourse-bound-unknown}
  $|Z| \leq 2 k$.
\end{restatable}

Now let $Y := W \setminus (Z \cup \{0\})$ be the nodes corresponding to edges that have not been removed and \begin{align*} \textstyle \bar{Y} := \{i \in Y : T[i] \setminus \{0, i\} \subseteq Z\} \end{align*} the nodes that correspond to maximal intervals that have not been removed.

The following lemma is a consequence of the exponential decay property. It shows that in order to establish a bound on the cost of $\algt$, it is enough to bound the cost of all paths $P^{(i)}$ for $i \in \bar{Y}$.

\begin{restatable}{lemma}{restateLemGenericBound}
\label{lem:bound-against-barY}
  $c(\algt) \leq c(\optt) + 3 \sum_{i \in \bar{Y}} c(P^{(i)})$
\end{restatable}

It remains to bound the cost of the paths associated with the tree nodes in $\bar{Y}$.
We establish a charging scheme by partitioning the line into three areas $A, B, C$:

\begin{enumerate}
  \item For $R \in \mathcal{R}$, let $A(R) := R \setminus P^{(j^*_R)}$. We define $A := \bigcup_{R \in \mathcal{R}} A(R)$.
  \item For $i \in \LIGHT$ and $\bar{R} \in \bar{\mathcal{R}}(i)$, let $B(\bar{R}) := \bar{R} \setminus \bigcup_{i' \in \LIGHT(i) \cap Z} P^{(i')}$.\\
  We define $B := \bigcup_{\bar{R} \in \bar{\mathcal{R}}} B(\bar{R})$.
  \item We define $C := \LINE \setminus (A \cup B)$.
\end{enumerate}

Consider a set $A(R)$ for some $R \in \mathcal{R}$. 
Recall that $i_R$ is the index of the smallest light interval constructed by \algUnknownSto containing $R$ and that $j^*_R$ is the first child interval of $i_R$ created by \algUnknownSto that intersects $R$.
From the choice of $j^*_R$ and the greedy construction of $P^{(j^*_R)}$ as a path of maximum $X^{(j^*_R)}$-gain we can conclude that $A(R)$ is not $\opto$-heavy; see \cref{fig:intuition-stage2} for an illustration.
Therefore $c(A(R) \setminus \opto) > \frac{1}{3} c(A(R))$.
Note that $\optt \cap A(R) = A(R) \setminus \opto$, because $A(R) \subseteq R$. 
Hence we obtain the following lemma.

\begin{restatable}{lemma}{restateLemHeavyIntervalsBudget}
\label{lem:heavy-intervals-budget}
  Let $R \in \mathcal{R}$. Then $\frac{1}{3} c(A(R)) \leq c(\optt \cap A(R))$.
\end{restatable}

A similar argument implies the same bound for all sets of the type $B(\bar{R})$ for some $\bar{R} \in \bar{\mathcal{R}}(i)$ and $i \in \LIGHT$.

\begin{restatable}{lemma}{restateLemLightIntervalsBudget}\label{lem:light-intervals-budget}
\label{lem:light-intervals-budget}
  Let $\bar{R} \in \bar{\mathcal{R}}$. Then $\frac{1}{3} c(B(\bar{R})) \leq c(\optt \cap B(\bar{R}))$.
\end{restatable}

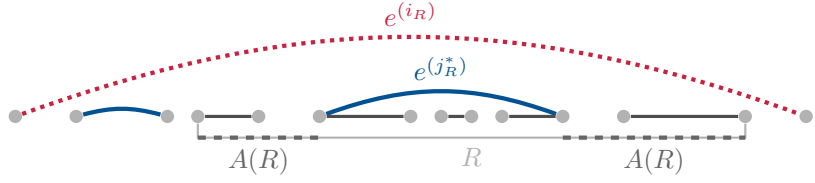
\begin{figure}[t]
\centering
      \begin{tikzpicture}[scale=.4]
      \foreach \pos/\name/\lab in { 
      {(0,1)/v0}, 
    {(3,1)/v1}, 
    {(4,1)/v2}, 
    {(6,1)/v3}, 
    {(8,1)/v4}, 
    {(16,1)/v5}, 
    {(18,1)/v6}, 
    {(22,1)/v9}, 
    {(-2,1)/vp1}, 
    {(24,1)/vp2}}
         \node[vertex,minimum size=5, color=black!30](\name) at \pos {};
    \foreach \source/ \dest  in {vp1/vp2}    
        \draw [TUMRot, ultra thick, dotted, above] (\source) to[bend left=20] node {$e^{(i_R)}$} (\dest);        
        \foreach \source/ \dest  in {{v0/v1}}
        \draw[TUMBlau, ultra thick, bend left=15] (\source) to node{} (\dest);
        
                \foreach \source/ \dest  in {{v4/v5}}    
        \draw[TUMBlau, ultra thick, bend left=20, above] (\source) to node{$e^{(j^*_R)}$} (\dest);
        
      \foreach \pos/\name/\lab in { 
      {(11,1)/c0}, 
    {(12,1)/c1}, 
    {(13,1)/c2}, 
    {(14,1)/c3}}
         \node[vertex,minimum size=5, color=black!30](\name) at \pos {};        
       
        \foreach \source/ \dest  in {{v2/v3}, {v4/c0}, {c1/c2}, {c3/v5}, {v6/v9}}
        \draw [color=black!70, very thick] (\source) to node {} (\dest); 
           
        \draw [black!30, thick,below] (v2) -- ++(0,-2em) -- node[below] {$R$} ($(v9) - (0,2em)$) -- (v9); 
        
               \foreach \source/ \dest  in {{v2/v4},{v5/v9}}
        \draw [black!60, ultra thick, dashed, below] ($(\source) - (0,2em)$) to node {$A(R)$} ($(\dest) - (0,2em)$); 
    \end{tikzpicture}
    \vspace{-0.4cm}
\caption{
Illustration of the proof of \cref{lem:heavy-intervals-budget}.
The solid dark grey lines depict edges in $\opto \cap R$.
At the time when \algUnknownSto constructed $P^{(j^*_R)}$, no other child interval of $i_R$ intersecting with $R$ was present.
Thus $X^{(j^*_R)} \cap R = X^{(i_R + 1)} \cap R$ = $\opto \cap R$.
If $A(R)$ was $\opto$-heavy, then $P^{(j^*_R)} \cup A(R)$ would be an $\opto$-heavy path of higher $\opto$-gain than~$P^{(j^*_R)}$, contradicting the greedy construction.} \label{fig:intuition-stage2}
\end{figure}

 Furthermore, one can show that the sets of the form $A(R)$ and $B(\bar{R})$ and the set $C$ form a partition of $L$ (see \cref{lem:line-partition} in the appendix). 
 We define $x : \LINE \rightarrow \mathbb{R}_+$ by 
 \begin{align*}
   x(e) := \begin{cases}
\frac{1}{3} & \text{if } e \in A \cup B\\
1 & \text{if } e \in C \cap \optt\\
0 & \text{if } e \in C \setminus \optt
\end{cases}
 \end{align*}
and obtain the following lemma as a consequence of \cref{lem:light-intervals-budget,lem:heavy-intervals-budget}.
\begin{restatable}{lemma}{restateLemXBoundedOptt}
\label{lem:x-bounded-optt}
$\sum_{e \in \LINE} c(e)x(e) \leq c(\optt \cap L)$.
\end{restatable}

We are now able to bound the cost of each path $P^{(i)}$ for $i \in \bar{Y}$ against its local budget $\sum_{e \in P^{(i)}} c(e)x(e)$. We consider the cases where $i \in \bar{Y}$ corresponds to a heavy and light edge in \cref{lem:x-beneath-heavy} and \cref{lem:x-beneath-light}, respectively.

  \begin{lemma}
\label{lem:x-beneath-heavy}
Let $j \in \bar{Y} \cap \HEAVY$. Then $c(P^{(j)}) \leq 6 \sum_{e \in P^{(j)}} c(e)x(e)$.     
\end{lemma}

\begin{proof}
  We first show that each request intersecting with $P^{(j)}$ is either contained in a child interval of $j$ or the edges in $P^{(j)}$ intersecting with the request are covered by $A$.
  \begin{claim*}
    Let $R \in \mathcal{R}$ with $R \cap P^{(j)} \neq \emptyset$. Then $R \cap P^{(j)} \subseteq A(R)$ or $R \subseteq P^{(i')}$ for some $i' \in \child{j}$.
  \end{claim*}
  \begin{claimproof}
    Assume $R \not\subseteq P^{(i')}$ for any $i' \in \child{j}$. 
    In particular, this implies that $i_R \notin \descendant{j}$.
    Let $i$ be the parent node of $j$ in $T$.
    We first exclude the possibility that $i_{R}$ is an ancestor of $i$. Indeed, if this was the case then $R \not\subseteq P^{(i)}$ and thus by \cref{ass:request-light-laminarity}, $P^{(i)} \subseteq R$.
    But then $P^{(i)}$ can neither contain other request intervals, nor can it intersect any non-request intervals.
    Therefore $\mathcal{R}(i) = \emptyset$ and $i \neq i^*_{\bar{R}}$ for all $\bar{R} \in \bigcup_{i' \in \LIGHT}\bar{\mathcal{R}}(i')$.
    Therefore, $i \notin Z$, a contradiction to $j \in \bar{Y}$.
    This implies that $i_R = i$.
    Further note that $j \in \bar{Y}$ implies $j \neq j^*_R$ for all $R \in \mathcal{R}$, as otherwise, $j$ would have been tagged for removal.
    We conclude that $j^*_R \in \child{i} \setminus \{j\}$ and thus $R \cap P^{(j)} \subseteq R \setminus P^{(j^*_R)} = A(R)$, as $P^{(j)}$ and $P^{(j^*_R)}$ are disjoint. 
  \end{claimproof}
   
  Let $Q := \bigcup_{i' \in \child{j}} P^{(i')}$.
  Consider $e \in P^{(j)} \cap \opto$. Note that the claim implies that if $e \notin A \cup Q$, then $e \in \optt$.
 We conclude that $(P^{(j)} \cap \opto) \setminus Q \subseteq A \cup B \cup (C \cap \optt)$.
  Further note that $c(P^{(j)} \cap \opto) \leq \frac{4}{3} c((P^{(j)} \cap \opto) \setminus Q)$ by \cref{lem:expo-decay-opto}. We obtain
 \begin{align*}
    \textstyle c(P^{(j)})  \leq \frac{3}{2} c(P^{(j)} \cap \opto) \leq 2 c((P^{(j)} \cap \opto) \setminus Q) \leq 6 \sum_{e \in P^{(j)}} c(e) x(e),
  \end{align*}
   where the first inequality follow from the fact that $P^{(j)}$ is $\opto$-heavy and the last inequality follows from the fact that $x(e) \geq 1/3$ for every $e \in A \cup B \cup (C \cap \optt)$. This proves the lemma.
\end{proof}

\begin{restatable}{lemma}{restateLemXBeneathLight}
\label{lem:x-beneath-light}
Let $i \in \bar{Y} \cap \LIGHT$. Then $c(P^{(i)}) \leq 3 \sum_{e \in P^{(i)}} c(e)x(e)$.
\end{restatable}

The proof of \cref{lem:x-beneath-light} is given in the appendix. We state a short proof sketch for intuition. Let $i \in \bar{Y} \cap \LIGHT$. Due to the removal of light edges with associated requests and \cref{ass:request-light-laminarity}, we know that $P^{(i)}$ either does not intersect with any request or $P^{(i)}$ is completely covered by a request. In the former case one can show that $P^{(i)}$ is included in $B$ and hence the lemma holds. In the latter case, $P^{(i)}$ is either included in $A$ and hence the lemma holds or it is included in $C$ in which case the lemma holds since $P^{(i)}$ is $O_1$-light. 

As a consequence, we can now show a first constant factor bound. Because the paths $P^{(i)}$ for $i \in \bar{Y}$ are pairwise disjoint, \cref{lem:x-bounded-optt,lem:x-beneath-heavy,lem:x-beneath-light} imply $\sum_{i \in \bar{Y}} c(P^{(i)}) \leq 6 c(\optt)$.
Plugging this into \cref{lem:bound-against-barY} we obtain 
\begin{align*}
\textstyle c(\algt) \leq c(\optt) + 3 \sum_{i \in \bar{Y}} c(P^{(i)}) \leq 19c(\optt).
\end{align*}

\subparagraph*{Improvement of approximation factor}

We can improve the approximation factor from $19$ to $10$ by a slight modification of the set of edges removed in Stage~2.
To this end, note that the factor for the bound given in \cref{lem:x-beneath-heavy} is greater than the one given in \cref{lem:x-beneath-light}. Indeed, the only reason for the weaker bound is that edges in $\bar{Y} \cap \HEAVY$ can have descendants with associated requests. Excluding this case improves the factor within the bound of the lemma from $6$ to $3$. We formalize this in the following lemma.

\begin{restatable}{lemma}{restateLemXBeneathHeavyNoChildren}
\label{lem:x-beneath-heavy-no-children}
Let $j \in \bar{Y} \cap \HEAVY$ such that $\child{j} \cap Z = \emptyset$. Then $c(P^{(j)}) \leq 3 \sum_{e \in P^{(j)}} c(e)x(e)$.
\end{restatable}

We now modify \algUnknownStt as follows: Compute the set $Z$ of edges tagged for removal by \algUnknownStt. Now construct the set $Z'$ by defining 
$$\textstyle \bar{H} := \{j \in \HEAVY \setminus Z \,:\, \child{j} \cap Z \neq \emptyset\} \qquad \text{and} \qquad
Z' := (Z \cup \bar{H}) \setminus \bigcup_{j \in \bar{H}} \child{j}.$$
Now execute step 2 of \algUnknownStt with $Z'$ instead of $Z$, i.e., remove the edges with indices in $Z'$ and connect the unmatched vertices by a min-cost matching.
It is easy to see that $|Z'| \leq |Z|$ and therefore the recourse factor is still bounded by $2$. For analyzing the approximation factor, we define $Y' := W \setminus (Z' \cup \{0\})$ the nodes corresponding to edges that have not been removed and $\bar{Y}' := \{i \in Y : T[i] \setminus \{0, i\} \subseteq Z'\}$ in analogy to the original analysis. It is easy to see that $\child{j} \cap Z = \emptyset$ for every $j \in \bar{Y}' \cap \HEAVY = \bar{Y} \setminus \bar{H}$ and hence $c(P^{(j)}) \leq 3 \sum_{e \in P^{(j)}} c(e)x(e)$ by \cref{lem:x-beneath-heavy-no-children}. Furthermore, if $i \in \bar{Y}' \cap \LIGHT$ then either $i \in \bar{Y}$ and $c(P^{(i)}) \leq 3 \sum_{e \in P^{(i)}} c(e)x(e)$ by \cref{lem:x-beneath-light}, or $i \in \child{j}$ for some $j \in \bar{Y} \setminus \bar{Y}' = \bar{H}$.
Note that \cref{lem:exp-decay,lem:x-beneath-heavy} imply
\begin{align*}
\textstyle \sum_{j \in \bar{H}} \sum_{i \in \child{j}} c(P^{(i)}) \leq \frac{1}{2} \sum_{j \in \bar{H}} c(P^{(j)}) \leq 3 \sum_{j \in \bar{H}} \sum_{e \in P^{(j)}} c(e)x(e).
\end{align*}
We thus obtain $\sum_{i \in \bar{Y}'} c(P^{(i)}) \leq 3 c(\optt)$ and hence for the modified algorithm it holds that
\begin{align*}
\textstyle c(\algt) \leq c(\optt) + 3 \sum_{i \in \bar{Y}'} c(P^{(i)}) \leq 10c(\optt).
\end{align*}
\bibliography{literature}

\newpage

\appendix

\section{Omitted proofs from \cref{sec:unknown-stage1}}

\subsection{Proof of \cref{lem:structure-opt-ontheline}}

  \restateLemOptOntheLine*
  \begin{proof}
    Enumerate the edges of $\opto$ by $e_1, \dots, e_n$ in an arbitrary order. 
    Then let $Q := \PL{e_1} \Delta \dots \Delta \PL{e_n}$, i.e., the symmetric difference of the projections of the edges to the line.
    Note that $Q \subseteq \LINE$ by construction and that $Q$ is a $V(\gro)$-join because $\PL{\{v, w\}}$ is a $\{v, w\}$-join for each edge $\{v, w\} \in \opto$.
    Because $\LINE$ is a path, it contains only a unique $V(\gro)$-join, namely the matching consisting of every odd edge of $\LINE$.
    Furthermore $Q \subseteq \bigcup_{i=1}^{n} \PL{e_i}$.
    Hence $c(Q) \leq \sum_{i=1}^{n} c(\PL{e_i}) = c(\opto)$, with strict inequality if $\PL{e_i} \cap \PL{e_j} \neq \emptyset$ for some $i \neq j$.
    We conclude that $Q = \opto$. 
  \end{proof}

\subsection{Proof of \cref{lem:heavy-union,lem:gain-maximizers}}
\label{apx:laminarity}

\restateLemHeavyUnion*

\begin{proof}
\begin{enumerate}
\item{
Let $X \subseteq E(\gro)$ and $A,B \subseteq E(\gro)$ be $X$-heavy sets with $A \cap B = \emptyset$. It follows that
\begin{align*}
c((A \cup B) \cap X) & = c(A \cap X) + c(B \cap X) \\ 
& \geq 2 \cdot c(A \setminus X) + 2 \cdot c(B \setminus X) = 2 \cdot c((A \cup B) \setminus X).
\end{align*}
The proof for two $X$-light sets follows analogously. 
}
\item{
Let $X \subseteq E(\gro)$ and $A,B \subseteq E(\gro)$ such that $B \subseteq A$, $A$ is $X$-heavy and $\g[X]{B} < 0$. Assume for contradiction that $A \setminus B$ is not $X$-heavy, that is, $c((A \setminus B) \cap X) < 2 \cdot c((A \setminus B) \setminus X)$. We derive a contradiction to $A$ being $X$-heavy, that is, 
\begin{align*}
c(A \cap X) & = c((A \setminus B) \cap X) + c(B \cap X) < c((A \setminus B) \cap X) + c(B \setminus X) \\ 
& < 2c((A \setminus B) \setminus X) + c(B \setminus X) \leq 2 c(A \setminus X).
\end{align*}
The first inequality follows from $\g[X]{B} < 0$ and the second from the assumption that $A \setminus B$ is not heavy. The proof of the second part of the lemma proceeds analogously. \qedhere
}
\end{enumerate}
\end{proof}

\restateLemGainMaximizers*
\begin{proof}
\begin{enumerate}
  \item Assume for contradiction that $X \subseteq L$, $X$ covers all vertices in $V(\PL{P})$ but $P \neq \PL{P}$. We show that $\PL{P}$ is then a $X$-alternating $X$-heavy path with higher gain than $P$.  
Note that $\PL{P}$ is $X$-alternating since $X$ covers all vertices in $V(\PL{P})$. Moreover, it holds that $c(P \cap X) \leq c(\PL{P} \cap X)$ since $X \subseteq L$ and it holds that $c(P \setminus X) > c(\PL{P} \setminus X)$. Hence, $\PL{P}$ is $X$-heavy and $\g[X]{P} < \g[X]{\PL{P}}$, a contradiction to the maximality of $P$.
  \item Let $X \subseteq \LINE$ be a matching and $I \subseteq L$ be an interval. Consider $P: = (I \cap X) \cup \{e(P') : P' \text{ maximal path in } I \setminus X\}$. Then, $P$ is $X$-alternating, $c(P \cap X) = c(I \cap X)$ and $c(P \setminus X) = c(\{e(P') : P \text{ maximal path in } I \setminus X\}) = c(I \setminus X)$. \qedhere
\end{enumerate}
\end{proof}

\subsection{Prefixes of gain-maximizing paths have positive gain}
In the proof of \cref{lem:laminarity} we used the fact that prefixes of gain-maximizing paths have positive gain, for which we give a formal proof here.
A \textit{prefix} of a path $P \subseteq E(\gro)$ is a non-empty subset $Q \subseteq P$ such that $Q$ is a path and $P \setminus Q$ is a path.

\begin{apxlemma} \label{lem:prefix}
Let $X \subseteq E(\gro)$ be a matching and $P \subseteq E(\gro)$ be a $X$-heavy, $X$-alternating path that is a maximizer of $\g[X]{P}$. Let $Q$ be a prefix of $P$. Then, $\g[X]{Q} \geq 0$.
\end{apxlemma}

\begin{proof}
Assume by contradiction that $\g[X]{Q} < 0$. Since $Q$ is a prefix of the $X$-alternating path $P$, $P' := P \setminus Q$ is a $X$-alternating path. Moreover, with \cref{lem:heavy-setminus} we know that $P'$ is $X$-heavy. Lastly, we obtain $\g[X]{P'} = \g[X]{P} - \g[X]{Q} > \g[X]{P}$. This yields a contradiction to the gain-maximality of $P$. 
\end{proof}

  \section{Omitted Proofs from \cref{sec:unknown-stage2}}
  
  \subsection{Proof of \cref{lem:request-structure}}
  
  \restateLemRequestStructure*
  
  \begin{proof}
    We first show that for every $\{u, v\} \in \optt$, either $\{u, v\} \in \LINE$ or $\{u, v\} \not\subseteq V(\gro)$.
    By contradiction assume there is an edge $\{u, v\} \in \optt \setminus \LINE$ but $u, v \in V(\gro)$.
    Because $\{u, v\} \notin \LINE$, there is $v' \in V(\PL{\{u,v\}}) \setminus \{u, v\}$.
    Let $u'$ be the matching partner of $v'$ in $\optt$.
    Let $v_1, \dots, v_4$ be an ordering of $\{u, v, u', v\}$ such that $v_1 < \dots < v_4$.
    Note that $[u, v] \cap [u', v'] \neq \emptyset$ and hence $c(u, v) + c(u', v') > c(v_1, v_2) + c(v_3, v_4)$.
    Thus $\optt \setminus \{\{u, v\}, \{u', v'\}\} \cup \{\{v_1, v_2\}, \{v_3, v_4\}\}$ is a matching of lower cost than $\optt$, a contradiction.
    
    Now let $P$ be a connected component of $\opto \Delta \optt$.
    Note that $P$ cannot be a cycle, because then $V(P) \subseteq V(\gro)$ but $P \not\subseteq L$, contradicting the observation above.
    Thus $P$ is a path starting and ending with an edge of $\optt$.
    Because it is $\opto$-alternating, every internal vertex of $P$ is in $V(\gro)$.
    Hence $P \cap \LINE$ contains all of $P$ except for its first and last edge.
    Therefore, $P \cap \LINE$ is the only request intersecting with $P$ and it starts and ends with an edge of $\opto$.
    This proves the lemma, because every request has to intersect with a connected component of $\opto \Delta \optt$ and there are only $k$ of them.
    \end{proof}

  \subsection{Proof of \cref{lem:i-and-j-exist}}
  
  \restateLemIandJexist*
   
    \begin{proof}
    Let $R \in \mathcal{R}(i)$. By \cref{lem:request-structure}, there must be an edge $e \in R \cap \opto$.
    Because $P^{(i)} \cap X^{(i)} = P^{(i)} \setminus \opto$ by \cref{lem:heavy-light}, the path $P^{(i)}$ starts and ends with an edge of $L \setminus \opto$.
    Therefore $P^{(i)} \cap \opto \subseteq X^{(i+1)}$ and hence $P^{(i)} \cap \opto \subseteq \bigcup_{j \in \child{i}} P^{(j)}$ by construction through \algUnknownSto.
    Hence $e \in P^{(j)} \cap R$ for some $j \in \child{i} \subseteq \HEAVY(i)$.
    
    Let $\bar{R} \in \bar{\mathcal{R}}(i)$. Since $\bar{R}$ is a maximal path in $P^{(i)} \setminus \bigcup_{R \in \mathcal{R}(i)} R$ and not a prefix of $P^{(i)}$, there are two requests $R', R'' \in \mathcal{R}(i)$ neighboring $\bar{R}$ (i.e., $R'$ and $R''$ each have exactly one endpoint with $\bar{R}$ in common). 
    By \cref{lem:request-structure}, $R'$ and $R''$ both start and end with edges of $\opto$.
    Therefore $\bar{R}$ starts and ends with edges in $\LINE \setminus \opto$, so in particular there is an $e \in \bar{R} \cap (\LINE \setminus \opto)$.
    Let $j \in \child{i}$ be the unique child such that $\bar{R} \subseteq P^{(j)}$.
    Analogously to the argument for $P^{(i)} \cap \opto$ given above, $P^{(j)} \setminus \opto \subseteq \bigcup_{i' \in \child{j}} P^{(i')}$. Hence $e \in P^{(i')}$ for some $i' \in \child{j} \subseteq \LIGHT(i)$. 
  \end{proof}

  \subsection{Proof of \cref{lem:recourse-bound-unknown}}  

\begin{apxlemma}\label{lem:non-request-counting}
  Let $i \in \LIGHT$. If $\mathcal{R}(i) \neq \emptyset$ then $|\bar{\mathcal{R}}(i)| \leq |\mathcal{R}(i)| - 1$.
\end{apxlemma}
\begin{proof}
Note that there are exactly $\vert \mathcal{R}(i) \vert + 1$ maximal paths in $P^{(i)} \setminus \bigcup_{R \in \mathcal{R}(i)} R$. This is because requests $R \in \mathcal{R}(i)$ start and end with edges from $\opto$ and are in particular not prefixes of the interval $P^{(i)}$. Hence, there are two maximal paths in $P^{(i)} \setminus \bigcup_{R \in \mathcal{R}(i)} R$ that are prefixes of $P^{(i)}$. In particular these cannot be subsets of $P^{(j)}$ for any $j \in ch(i)$. Summarizing, there are at most $\vert \mathcal{R}(i) \vert -1$ intervals that fulfill both conditions in the definition of $\bar{\mathcal{R}}(i)$.
\end{proof}

  \restateLemRecourseBoundedUnknown*
  
  \begin{proof}
  Let $Z^{(i)}$ be the set of edge indices tagged for removal when processing node $i \in L$.
  We show $|Z^{(i)}| \leq 2|\mathcal{R}(i)|$, proving the lemma.
  If $\mathcal{R}(i) = \emptyset$, then $Z^{(i)} = \emptyset$.
  If $\mathcal{R}(i) \neq \emptyset$, then $|Z^{(i)}| \leq |\mathcal{R}(i)| + |\bar{\mathcal{R}}(i)| + 1 \leq 2|\mathcal{R}(i)|$, where the last inequality is due to \cref{lem:non-request-counting}.  
\end{proof}

  \subsection{Proof of \cref{lem:bound-against-barY}}
  
 For proving the following lemmas, we formalize the notion of the set of descendants of a tree-node $i \in W$ by defining $\descendant{i} := \{j \in W : i \in V(T[j]), j \neq i\}$.  
  
  \restateLemGenericBound*
  \begin{proof}
  Recall that $\algt = M' \cup M''$, where $M' = \{e^{(i)} : i \in Y)\}$ and $M''$ is a minimum cost perfect matching on the set of vertices $$U := \{v \in V(\grt) : v \text{ is not covered by } M'\}.$$
  
  Let $i_1, \dots, i_\ell$ be an arbitrary ordering of the indices in $Y$ and consider the symmetric difference of paths $P^{(i)}$ for $i \in Y$, i.e., \mbox{$\bar{M} := P^{(i_1)} \Delta \dots \Delta P^{(i_\ell)}$}.
  Note that $\bar{M}$ is a $(V(\grt) \setminus U)$-join (as each path in $P^{(i)}$ for $i \in Y$ corresponds to an edge in $M'$) and that $\bar{M} \subseteq \bigcup_{i \in \bar{Y}} P^{(i)}$ by construction of $\bar{Y}$.
  Thus $c(\bar{M}) \leq \sum_{i \in \bar{Y}} c(P^{(i)})$.
  
  Because $\optt$ is a $V(\grt)$-join, $\bar{M} \Delta \optt$ is a $U$-join.
  Therefore $c(M'') \leq c(\bar{M} \Delta \optt) \leq c(\bar{M}) + c(\optt)$.
  Furthermore, 
  $$c(M') = \sum_{i \in Y} c(P^{(i)}) \leq \sum_{i \in \bar{Y}} \left(c(P^{(i)}) + \sum_{j \in \descendant{i}} c(P^{(j)})\right) \leq 2 \sum_{i \in \bar{Y}} c(P^{(i)}).$$
  The last inequality follows from the exponential decay property established in \cref{lem:exp-decay}.
  Putting this together, we obtain $c(\algt) = c(M') + c(M'') \leq c(\optt) + c(M') + c(\bar{M}) \leq c(\optt) + 3 \sum_{i \in \bar{Y}} c(P^{(i)})$.  
\end{proof}

\subsection{Proofs of \cref{lem:heavy-intervals-budget,lem:light-intervals-budget}}

\begin{apxlemma}
\label{lem:greedy-removal}
  Let $i \in W$ and $R \subseteq P^{(i)}$ be an interval. Let $Q := R \setminus P^{(j^*)}$ where
  $j^* := \min \, \{j \in \child{i} : P^{(j)} \cap R \neq \emptyset\}$.
  If $i \in \LIGHT$, then $Q$ is not $\opto$-heavy.
  If $i \in \HEAVY$, then $Q$ is not $\opto$-light.
\end{apxlemma}

\begin{proof}
  By contradiction assume that $i \in \LIGHT$ and $Q$ is $\opto$-heavy.
  Note that $P^{(j)} \cap R = \emptyset$ for all $j \in \child{i}$ with $j < j^*$ by choice of $j^*$. 
  This implies that $X^{(j^*)} \cap R = (X^{(i)} \Delta P^{(i)}) \cap R = \opto \cap R$, since $i \in \LIGHT$.
  In particular, $Q$ is $X^{(j^*)}$-heavy.
  Since $Q$ is $X^{(j^*)}$-heavy, $P^{(j^*)} \cup Q$ is an $X^{(j^*)}$-heavy $X^{(j^*)}$-alternating path by \cref{lem:heavy-union}.
  Furthermore, $\g[X^{(j^*)}]{Q \cup P^{(j^*)}} = \g[X^{(j^*)}]{Q} + \g[X^{(j^*)}]{P^{(j^*)}}$, contradicting the construction of $P^{(j^*)}$ by \algUnknownSto. 
  The proof for the case that $i \in \HEAVY$ and $Q$ is $\opto$-light follows analogously.
\end{proof}

  In the proof of \cref{lem:light-intervals-budget}, we make use of the following consequence of \cref{ass:request-light-laminarity,ass:request-heavy-prefixes}.
  
\begin{apxlemma}\label{lem:no-requests-in-children}
  Let $i \in \LIGHT$. If $\mathcal{R}(i) = \emptyset$, then $\mathcal{R}(i') = \emptyset$ for all $i' \in \descendant{i} \cap \LIGHT$.
\end{apxlemma}
\begin{proof}
  We show the contrapositive, i.e., if $\mathcal{R}(i') \neq \emptyset$ for some $i' \in \descendant{i} \cap \LIGHT$, then $\mathcal{R}(i) \neq \emptyset$. 
  Let $i' \in \descendant{i} \cap \LIGHT$ and $R \in \mathcal{R}(i')$. 
  Note that because $i' \in \LIGHT$, there must be $j \in \child{i}$ with $i' \in \descendant{j}$. 
  Hence $R \cap P^{(j)} \supseteq R \cap P^{(i')} \neq \emptyset$. 
  By \cref{ass:request-heavy-prefixes}, there is a $R' \in \mathcal{R}$ containing the first edge $e$ of $P^{(j)}$.
  Note that $e \in P^{(i)}$ but $e$ is not contained in any child interval of $j$, because all intervals created by \algUnknownSto have disjoint endpoints. 
  This implies that $i_{R'} \leq i$, because $e \in R' \subseteq P^{(i_{R'})}$.
  If $i_{R'} = i$, then $\mathcal{R}(i) \neq \emptyset$ and we are done.
  If $i_{R'} < i$, then $R' \not\subseteq P^{(i)}$ and hence, by \cref{ass:request-light-laminarity}, $P^{(i)} \subseteq R'$.
  But then $P^{(i')} \subseteq R'$, which implies $\mathcal{R}(i') = \emptyset$ as request intervals are disjoint, a contradiction.  
\end{proof}

\restateLemHeavyIntervalsBudget*
\begin{proof}
  \cref{lem:greedy-removal} implies that $A(R) = R \setminus P^{(j^*_R)}$ is not $\opto$-heavy.
  Note that $A(R) \setminus \optt = A(R) \cap \opto$ because $A(R) \subseteq R \subseteq \opto \Delta \optt$.
  Thus $c(A(R) \setminus \optt) = c(A(R) \cap \opto) < 2 c(A(R) \setminus \opto) = 2 c(A(R) \cap \optt)$.
\end{proof}

\restateLemLightIntervalsBudget*

\begin{proof}
Let $\bar{R} \in \bar{\mathcal{R}}$ and $i$ be the unique index such that $\bar{R} \in \mathcal{\bar{R}}(i)$. 
Let $i^*_{\bar{R}} := \min \{i' \in \LIGHT(i) : P^{(i')} \cap \bar{R} \neq \emptyset \}$, i.e., the index of the light edge that was removed because of $\bar{R}$. Then, because of \cref{lem:greedy-removal}, $\bar{R} \setminus P^{(i^*_{\bar{R}})}$ is not $\opto$-light. We obtain $B(\bar{R})$ from $\bar{R} \setminus P^{(i^*_{\bar{R}})}$ by deleting edges from $\bigcup_{i' \in \LIGHT(i) \cap Z} P^{(i')}$ that intersect with $\bar{R}$. By \cref{ass:request-light-laminarity} we know that $(\bar{R} \setminus P^{(i^*_{\bar{R}})}) \cap \bigcup_{i' \in \LIGHT(i) \cap Z} P^{(i')}$ is the union of $\opto$-light intervals and hence $\opto$-light (\cref{lem:heavy-union}). We conclude with \cref{lem:heavy-setminus} that $B(\bar{R})$ is not $\opto$-light. Note that by construction $B(\bar{R})$ does not intersect with requests corresponding to ancestors of $i$ or with children of $i$ that contain requests. Lastly, with help of \cref{lem:no-requests-in-children} we can deduce that $B(\bar{R})$ does not intersect with the request of any descendant of $i$. We obtain that $\opto \cap B(\bar{R}) = \optt \cap B(\bar{R})$ and from $B(\bar{R})$ being not $\opto$-light we get that $B(\bar{R})$ is not $\optt$-light. This concludes the proof.
\end{proof}

\subsection{Proof of \cref{lem:x-bounded-optt}}

\begin{apxlemma}
\label{lem:line-partition}
  $\LINE = \dot{\bigcup}_{R \in \mathcal{R}} A(R) \, \dot\cup \, \dot{\bigcup}_{\bar{R} \in \mathcal{\bar{R}}} B(\bar{R}) \, \dot\cup \, C$
\end{apxlemma}

\begin{proof}
We show $A(R) \cap B(\bar{R}) = A(R) \cap A(R') = B(\bar{R}) \cap B(\bar{R}') = \emptyset$  for all $R,R' \in \mathcal{R}, \bar{R}, \bar{R}' \in \mathcal{\bar{R}}$, then the lemma follows directly from the definition of~$C$.

First, let $R \in \mathcal{R}$ and $\bar{R} \in \bar{\mathcal{R}}$.
By contradiction assume $R \cap B(\bar{R}) \neq \emptyset$.
Let~$i$ be such that $\bar{R} \in \bar{\mathcal{R}}(i)$.
Note that $R \cap B(\bar{R}) \subseteq R \cap P^{(i)}$ and hence $R$ intersects~$P^{(i)}$.
If $i_R \in \ancestor{i}$, then $R \not\subseteq P^{(i)}$ and hence $P^{(i)} \subseteq R$ by \cref{ass:request-light-laminarity}. This implies $\mathcal{R}(i) = \emptyset$ and hence $\bar{\mathcal{R}}(i) = \emptyset$, which contradicts $\bar{R} \in \bar{\mathcal{R}}(i)$.
If $i_R = i$, then $R \cap \bar{R} = \emptyset$ because $\bar{R} \subseteq P^{(i)} \setminus \bigcup_{R' \in \mathcal{R}(i)} R'$.
Thus $i_R \in  \descendant{i}$ and more specifically $i_R \in \descendant{j_{\bar{R}}}$, where $j_{\bar{R}}$ is the unique child of $i$ that intersects with $\bar{R}$.
But this implies $R \subseteq P^{(i')}$ for some $i' \in \child{j_{\bar{R}}}$.
Thus $\mathcal{R}(i') \neq \emptyset$ by \cref{lem:no-requests-in-children}, and thus $B(\bar{R}) \subseteq \bar{R} \setminus P^{(i')}$ by construction of $B(\bar{R})$.

Second, let $R,R' \in \mathcal{R}$. Then, $R \cap R' \neq \emptyset$ since requests are disjoint and hence $A(R) \cap A(R') \neq \emptyset$. 

Finally, let $\bar{R},\bar{R}' \in \mathcal{\bar{R}}$ and $i, i'$ such that $\bar{R} \in \mathcal{\bar{R}}(i)$ and $\bar{R}' \in \mathcal{\bar{R}}(i')$. It is easy to see that if $i'$ is not a descendant of $i$ or vice versa, then $\bar{R} \cap \bar{R}' = \emptyset$. Hence, assume w.l.o.g. assume that $i' \in \descendant{i}$. Then there exists a child of $i$, say $i'' \in \child{i}$, such that $P^{(i')} \subseteq P^{(i'')}$. Morever, from $\mathcal{\bar{R}}(i') \neq \emptyset$ we deduce that $\mathcal{R}(i') \neq \emptyset$ and due to \cref{lem:no-requests-in-children} it follows that $\mathcal{R}(i'') \neq \emptyset$. Therefore, $ B(\bar{R}) \cap P^{(i')}  = \emptyset$ which implies $B(\bar{R}) \cap B(\bar{R}') = \emptyset$. 
\end{proof}

\restateLemXBoundedOptt*
\begin{proof}
  An implication of \cref{lem:line-partition} is that $A, B, C$ are pairwise disjoint and hence
  \begin{align*}
    \sum_{e \in \optt} c(e) & \geq \sum_{e \in A \cap \optt} c(e) + \sum_{e \in B \cap \optt} c(e) + \sum_{e \in C \cap \optt} c(e)
  \end{align*}
  Moreover, \cref{lem:line-partition} states that all sets $A(R)$ for $R \in \mathcal{R}$ and all sets
  $B(\bar{R})$ for $\bar{R} \in \bar{\mathcal{R}}$ are pairwise disjoint and therefore \cref{lem:heavy-intervals-budget,lem:light-intervals-budget} imply that $\sum_{e \in (A \cup B) \cap \optt} c(e) \geq \frac{1}{3} \sum_{e \in A \cup B} c(e)$.
  The lemma follows from plugging in the definition of $x$.
\end{proof}

  \subsection{Consequence of exponential decay (used in proof of \cref{lem:x-beneath-heavy})}
  \label{apx:heavy-bound}

  \begin{apxlemma}\label{lem:expo-decay-opto}
    Let $i \in \HEAVY$. Then $\sum_{j \in \child{i}} c(P^{(j)} \cap \opto) \leq \frac{1}{4} \cdot c(P^{(i)} \cap \opto)$.
  \end{apxlemma}

\begin{proof}
  Let $i \in \HEAVY$. Then
  \begin{align*}
    \sum_{j \in \child{i}} \!\! c(P^{(j)} \cap \opto) \leq \frac{1}{2} \sum_{j \in \child{i}} c(P^{(j)} \setminus \opto) \leq \frac{1}{2} c(P^{(i)} \setminus \opto) \leq \frac{1}{4} c(P^{(i)} \cap \opto),
  \end{align*}
  where the first inequality follows from the fact that $P^{(j)}$ is $\opto$-light for all \mbox{$j \in \child{i} \subseteq \LIGHT$}, the second follows from the fact that the intervals $P^{(j)}$ of the children are disjoint and all contained in $P^{(i)}$, and the last inequality follows from the fact that $P^{(i)}$ is $\opto$-heavy.
\end{proof}
  
  \subsection{Proof of \cref{lem:x-beneath-light}}
\restateLemXBeneathLight*
\begin{proof}
  We distinguish two cases:
  \begin{description}
    \item[Case 1 ($P^{(i)} \cap \bigcup_{R \in \mathcal{R}} R  = \emptyset$):]
    Let $j' \in \HEAVY$ be the parent node of $i$ in $T$ and let $i' \in \LIGHT$ be the parent node of $j'$ in $T$.
    Note that $i \in \bar{Y}$ implies that $i', j' \in Z$.
    Since heavy nodes are only tagged for removal when their intervals intersect with requests, we conclude that $P^{(j')} \cap  \bigcup_{R \in \mathcal{R}(i')} R \neq \emptyset$.
    Because $P^{(i)} \subseteq P^{(j')} \setminus \bigcup_{R \in \mathcal{R}(i')} R$ and because $P^{(j')}$ starts and ends with a request edge by \cref{ass:request-heavy-prefixes}, there must be a non-request interval $\bar{R} \in \bar{\mathcal{R}}(i')$ with $P^{(i')} \subseteq \bar{R}$.
    Then $i \notin Z$ implies $P^{(i)} \subseteq B(\bar{R})$.
    Thus $P^{(i)} \subseteq B$ implies the statement of the lemma.
    
    \item[Case 2 ($P^{(i)} \cap R \neq \emptyset$ for some $R \in \mathcal{R}$):]  
    Note that $i \in Y$ implies $\mathcal{R}(i) = \emptyset$. 
    By \cref{lem:no-requests-in-children}, $i_R$ must be an ancestor of $i$ in $T$ and hence $R \not\subseteq P^{(i)}$.
    Thus $P^{(i)} \subseteq R$ by \cref{ass:request-light-laminarity}.
    \begin{itemize}
      \item  If $j^*_R$ is not an ancestor of $i$, then $P^{(i)} \subseteq R \setminus P^{(j^*_R)} = A(R) \subseteq A$ and we are done because $x(e) = 1/3$ for all $e \in A$.
    \item If $j^*_R$ is an ancestor of $i$, then $P^{(i)} \subseteq R \cap P^{(j^*_R)}$.
    Note that $R \cap P^{(j^*_R)} \subseteq C$, as $R \cap A(R') = \emptyset$ for all $R' \in \mathcal{R} \setminus \{R\}$ and $R \cap B(\bar{R}) = \emptyset$ for all $\bar{R} \in \bigcup_{i' \in \LIGHT} \bar{\mathcal{R}}(i')$.
    Thus $P^{(i)} \subseteq C$ in this latter case. 
    Then
    \begin{align*}
      \sum_{e \in P^{(i)}} c(e)x(e) & = c(P^{(i)} \cap \optt) = c(P^{(i)} \setminus \opto) \geq \frac{1}{3} c(P^{(i)}), 
    \end{align*}
    where the first equality follows from $P^{(i)} \subseteq C$, the second equality follows from $P^{(i)} \subseteq R$ and the final inequality follows from the fact that $P^{(i)}$ is $\opto$-light.\qedhere
    \end{itemize}
  \end{description}
\end{proof}

\subsection{Proof of \cref{lem:x-beneath-heavy-no-children}}\label{apx:subsec:noChildren}

\restateLemXBeneathHeavyNoChildren*

\begin{proof}

We distinguish two cases.

\begin{description}

\item[Case 1 (There is $R \in \mathcal{R}$ with $P^{(j)} \cap R \neq \emptyset$):]
We show that in this case $P^{(j)} \subseteq A(R)$.
First assume that $P^{(j)} \not\subseteq R$. Then by \cref{ass:request-heavy-prefixes} there are two distinct requests $R', R'' \in \mathcal{R}$ such that the first edges of $P^{(j)}$ is in $R'$ and its last edge is in $R''$. Further, by \cref{ass:request-light-laminarity}, $R', R'' \in \mathcal{R}(i)$, where $i$ is the parent node of $j$ in $T$.
Since $R' \neq R''$, there is a non-request interval $\bar{R} \in \bar{\mathcal{R}}(i)$ with $\bar{R} \subseteq P^{(j)}$. 
But then $i^*_{\bar{R}} \in \child{j} \cap Z$, a contradiction. We conclude that $P^{(j)} \subseteq R$.

Now recall the claim that we showed in the proof of \cref{lem:x-beneath-heavy}. This claim of course also holds true in context of the present lemma, whose premise is a special case of that of \cref{lem:x-beneath-heavy}. 
  \begin{claim*}
    Let $R \in \mathcal{R}$ with $R \cap P^{(j)} \neq \emptyset$. Then $R \cap P^{(j)} \subseteq A(R)$ or $R \subseteq P^{(i')}$ for some $i' \in \child{j}$.
  \end{claim*}
  Since $P^{(j)} \subseteq R$, the claim implies $P^{(j)} \subseteq A(R)$ and hence $x(e) \geq 1/3$ for all $e \in P^{(j)}$ by definition of $x$. This concludes the first case.
  
  \item[Case 2 ($P^{(j)} \cap R = \emptyset$ for all $R \in \mathcal{R}$):]
   Note that this implies that $P^{(j)} \subseteq L \setminus A = B \cup C$.
   We show that either $P^{(j)} \subseteq B$ or $P^{(j)} \subseteq C$.
   
   First assume $B \cap P^{(j)} \neq \emptyset$. Then there is $\bar{R} \in \bar{\mathcal{R}}$ with $P^{(j)} \cap B(\bar{R}) \neq \emptyset$. 
   Let $i'$ be such that $\bar{R} \in \bar{\mathcal{R}}(i')$.
   Since $P^{(j)}$ does not intersect with any request, $i$ must be an ancestor of $j$ in $T$ and $P^{(j)} \subseteq \bar{R}$. 
   Recall that $B(\bar{R}) = \bar{R} \setminus \bigcup_{i' \in \LIGHT(i) \cap Z} P^{(i')}$. Because $i' \notin Z$ for all $i' \in \child{j}$, either $P^{(j)} \subseteq B(\bar{R})$ or there is $i' \in \LIGHT(i) \cap Z$ such that $P^{(j)} \subseteq P^{(i')}$, implying $\cap B(\bar{R}) = \emptyset$. The latter yields a contradiction and hence  $P^{(j)} \subseteq B(\bar{R}) \subset B$.
   Therefore $x(e) \geq 1/3$ for all $e \in P^{(j)}$ by definition of $x$.
   
   Now assume $P^{(j)} \cap B = \emptyset$, which implies $P^{(j)} \subseteq C$. In this case
   $$c(P^{(j)}) \geq \frac{3}{2} c(P^{(j)} \cap \opto) = \frac{3}{2} c(P^{(j)} \cap \optt \cap C) = \frac{3}{2} \sum_{e \in P^{(j)}} c(e)x(e)$$
   where the first inequality follows from the fact that $P^{(j)}$ is $\opto$-heavy, the first equality follows from $R \cap P^{(j)} = \emptyset$ for all $R \in \mathcal{R}$, and the second equality follows from the definition of $x$.
   \end{description}
\end{proof}

\subsection{\cref{ass:request-light-laminarity,ass:request-heavy-prefixes} are without loss of generality}
\label{apx:assumptions-wlog}

If $\mathcal{R}$ violates \cref{ass:request-light-laminarity} or \cref{ass:request-heavy-prefixes}, we construct a modified set of requests as follows to obtain a modified request set $\mathcal{R}'$ which fulfills both assumptions and corresponds to second stage arrivals with lower cost than the actual $\optt$. We then run \algUnknownStt for this modified request set.

\medskip
\noindent\textsf{\textbf{Algorithm for Modifying $\mathcal{R}$}}\\
\textbf{(i)} Let $\mathcal{R}' := \mathcal{R}$.
\textbf{(ii)} While \cref{ass:request-light-laminarity} or \cref{ass:request-heavy-prefixes} is violated, apply one of the following rules:
\begin{enumerate}
  \item For $R \in \mathcal{R}'$ and $i \in \LIGHT$ with $R \setminus P^{(i)} = \emptyset$ and $P^{(i)} \setminus R \neq \emptyset$: Replace $R$ by $R' := R \setminus P^{(i)}$ (if $R' = \emptyset$, then remove $R$ entirely).
  \item For $j \in \HEAVY$ such that $P^{(j)} \cap \bigcup_{R \in \mathcal{R}'} R \neq \emptyset$ but  $P^{(j)}$ does not start and end with an edge of $\bigcup_{R \in \mathcal{R}'} R$: 
  Let $Q$ be the maximal prefix of $P^{(j)}$ not contained in $\bigcup_{R \in \mathcal{R}'} R$. Let $R$ be the request with $R \cap P^{(j)} \neq \emptyset$ that shares an endpoint with $Q$.
  Replace $R$ by $R' := R \cup Q$.
\end{enumerate}

It is easy to see that $|\mathcal{R}'| \leq |\mathcal{R}|$ and that $R_1 \cap R_2 \neq \emptyset$ for all $R_1, R_2 \in \mathcal{R}'$. 
We thus obtain a new set of requests $\mathcal{R}$' fulfilling \cref{ass:request-light-laminarity,ass:request-heavy-prefixes}.
We let $O' := \opto \Delta \bigcup_{R \in \mathcal{R}'} R$.

\begin{apxlemma}
$c(O') \leq c(\optt \cap \LINE)$.
\end{apxlemma}
\begin{proof}
  Consider an application of Rule~1 to $i \in \LIGHT$ and $R \in \mathcal{R}$.
  Let $\bar{R} := R \cap P^{(i)}$.
  Observe that $\bar{R}$ is a prefix of $P^{(i)}$ and therefore $\g[X^{(i)}]{\bar{R}} > 0$. 
  Note that $X^{(i)} \cap P^{(i)} = P^{(i)} \setminus \opto$ because $i \in \LIGHT$. 
  We conclude $c(\bar{R} \cap O') = c(\bar{R} \cap \opto) < c(\bar{R} \setminus \opto) = c(\bar{R} \cap \optt)$.
  
  Consider an application of Rule~2 to $j \in \HEAVY$ and $R \in \mathcal{R}$.
  Note that $\g[X^{(j)}]{Q} > 0$ because $Q$ is a prefix of $P^{(j)}$.
  Note further that $X^{(j)} \cap P^{(j)} = P^{(j)} \cap \opto$ because $j \in \HEAVY$.
  We conclude $c(Q \cap O') = c(Q \setminus \opto) < c(Q \cap \opto) = c(Q \cap \optt)$. 
  
  Hence, applications of Rule~1 or~2 do not increase the cost of $O'$ compared to $\optt$.
\end{proof}

We run \algUnknownStt for this modified set of requests, i.e., we use $\mathcal{R}'$ instead of $\mathcal{R}$ to determine the matching $M'$. As before, let $Z$ be the set of indices of the removed edges, let $Y := W \setminus (Z \cup \{0\})$, and $\bar{Y} := \{i \in Y : T[i] \setminus \{0, i\} \subseteq Z\}$.
Then $c(\algt) \leq c(\optt) + 3 \sum_{i \in \bar{Y}} c(P^{(i)})$ by \cref{lem:bound-against-barY} (note that the proof of the lemma does not make any assumptions about $Z$).
Furthermore, note that the proofs of \cref{lem:heavy-intervals-budget,lem:light-intervals-budget,lem:x-bounded-optt,lem:x-beneath-heavy,lem:x-beneath-light,lem:x-beneath-heavy-no-children} work without alteration when replacing $\mathcal{R}$ by $\mathcal{R}'$ and $\optt$ by $O'$.
We thus obtain
$\sum_{i \in \bar{Y}} c(P^{(i)}) \leq 6 c(O') \leq 6c(\optt \cap L)$ (or $\sum_{i \in \bar{Y}'} c(P^{(i)}) \leq 3c(\optt \cap L)$ in case of the improved version of \algUnknownStt).
Hence $c(\algt) \leq 19 c(\optt)$ (or $c(\algt) \leq 10 c(\optt)$ in case of the improved version of \algUnknownStt).

\end{document}